\documentclass[amsart,12pt,oneside,english]{article}
\usepackage{amsthm, amsmath, amsfonts, amsxtra, amssymb, euscript, mathrsfs, MnSymbol, verbatim, enumerate, multicol, multirow, color,babel, geometry, tikz,tikz-cd, tikz-3dplot, tkz-graph, array, enumitem, hyperref, thm-restate, thmtools, datetime, graphicx, tensor, braket, slashed, adjustbox, mathtools, bbding, esint, pgfplots, ytableau, float, rotating, mathdots, savesym, cite, wasysym, amscd, pifont, setspace, wrapfig, picture, subfigure, tabularx,youngtab,parskip}
\usepackage[numbers,sort&compress]{natbib}
\usepackage[normalem]{ulem} 
\usepackage[utf8]{inputenc}
\usepackage[all]{xy}
\numberwithin{equation}{section}
\numberwithin{figure}{section}
\usetikzlibrary{arrows, positioning, decorations.pathmorphing, decorations.markings, decorations.pathreplacing, decorations.markings, matrix, patterns, snakes}
\tikzset{big arrow/.style={
    decoration={markings,mark=at position 1 with {\arrow[scale=2,#1]{>}}},
    postaction={decorate},
    shorten >=0.4pt},
  big arrow/.default=black,
  ->-/.style={decoration={markings, mark=at position #1 with {\arrow[scale=2.5]{>}}},postaction={decorate}}}
  
\hypersetup{colorlinks=true}
\hypersetup{linkcolor=black}
\hypersetup{citecolor=black}
\hypersetup{urlcolor=black}

\makeatletter
\def\oversortoftilde#1{\mathop{\vbox{\m@th\ialign{##\crcr\noalign{\kern3\p@}%
				\sortoftildefill\crcr\noalign{\kern3\p@\nointerlineskip}%
				$\hfil\displaystyle{#1}\hfil$\crcr}}}\limits}

\def\sortoftildefill{$\m@th \setbox\z@\hbox{$-$}%
	\braceld\leaders\vrule \@height\ht\z@ \@depth\z@\hfill\braceru$}

\makeatother

\theoremstyle{plain}
\newtheorem*{thm*}{Theorem}
\newtheorem{thm}{Theorem}[section]

\theoremstyle{definition}
\newtheorem{defn}[thm]{Definition}
\newtheorem*{defn*}{Definition}

\newtheorem{rem}[thm]{Remark}

\DeclareMathOperator{\slim}{s-lim}
\DeclareMathOperator{\wlim}{w-lim}
\makeatother
\newcommand{\calo}{\mathcal{O}}
\newcommand{\cala}{\mathcal{A}}
\newcommand{\calp}{\mathcal{P}}
\newcommand{\cals}{\mathcal{S}}
\newcommand{\calh}{\mathcal{H}}
\newcommand{\calk}{\mathcal{K}}
\newcommand{\calb}{\mathcal{B}}
\newcommand{\trace}{\text{Tr }}

\begin{document}

\begin{titlepage}
\vspace*{-3cm} 
\begin{flushright}
{\tt CALT-TH-2019-042}\\
\end{flushright}
\begin{center}
\vspace{2.5cm}
{\LARGE\bfseries Entanglement Wedge Reconstruction of Infinite-dimensional von Neumann Algebras using Tensor Networks \\  }
\vspace{2cm}
{\large
Monica Jinwoo Kang$^{1}$ and David K. Kolchmeyer$^2$\\}
\vspace{.6cm}
{ $^1$ Walter Burke Institute for Theoretical Physics, California Institute of Technology}\par\vspace{-.3cm}
{  Pasadena, CA 91125, U.S.A.}\par
\vspace{.2cm}
{ $^2$ Department of Physics, Harvard University, Cambridge, MA, USA}
\vspace{.6cm}

\scalebox{.95}{\tt  monica@caltech.edu, dkolchmeyer@g.harvard.edu}\par
\vspace{2cm}
{\bf{Abstract}}\\
\end{center}
{ Quantum error correcting codes with finite-dimensional Hilbert spaces have yielded new insights on bulk reconstruction in AdS/CFT. In this paper, we give an explicit construction of a quantum error correcting code where the code and physical Hilbert spaces are infinite-dimensional. We define a von Neumann algebra of type II$_1$ acting on the code Hilbert space and show how it is mapped to a von Neumann algebra of type II$_1$ acting on the physical Hilbert space. This toy model demonstrates the equivalence of entanglement wedge reconstruction and the exact equality of bulk and boundary relative entropies in infinite-dimensional Hilbert spaces. }
\\
\vfill 
\end{titlepage}

\tableofcontents
\newpage

\section{Introduction}

The study of entanglement entropy has utilized results in the mathematical field of operator algebras \cite{Faulkner,LongoXu:1,LongoXu:2}. In quantum field theory, von Neumann algebras are associated with causally complete subregions of spacetime \cite{Haag}. Since AdS/CFT implies that information in the bulk is encoded redundantly in the boundary, quantum error correction is a natural framework in which to elucidate the connection between holographic quantum field theories and their gravity duals \cite{Harlow:2018fse,Almheiri:2014lwa,DongHarlowWall,HaydenQi}. Quantum error correction with finite-dimensional Hilbert spaces has been used to argue that entanglement wedge reconstruction is identical to the Ryu--Takayanagi formula and the equivalence of bulk and boundary relative entropies \cite{Harlow:2016fse,DongHarlowWall}. In order to study a more realistic toy model where boundary subregions are characterized by infinite-dimensional von Neumann algebras, we should consider quantum error correcting codes defined on infinite-dimensional Hilbert spaces.

The purpose of this paper is to construct a Quantum Error Correcting Code (QECC) where the physical Hilbert space and the code subspace are infinite-dimensional and admit the action of infinite-dimensional von Neumann algebras. We describe a toy model that allows us to see how a von Neumann algebra on the code subspace is reconstructed on the physical Hilbert space. Von Neumann algebras acting on finite-dimensional Hilbert spaces must be of type I. Our toy model contains an example of an infinite-dimensional von Neumann algebra, namely a type II$_1$ factor, which is defined and explained in Section \ref{sec:vNalgtypes}.\footnote{There are other types of von Neumann algebras that may act on infinite-dimensional Hilbert spaces. The local operator algebras that arise in quantum field theory are generically of type III$_1$ \cite{gabbani, Witten:2018zxz}.}

Furthermore, we show that in the context of operator-algebra quantum error correction, this QECC satisfies the following two statements:\vspace{-3pt}
\begin{itemize}
\item Entanglement wedge reconstruction \cite{HeadrickHubeny2014,Wall2012,Czech2012}, \vspace{-5pt}
\item Relative entropy equals bulk relative entropy (JLMS formula \cite{Jafferis:2015del}).
\end{itemize}\vspace{-3pt}
In particular, we first show that our QECC satisfies entanglement wedge reconstruction for a particular choice of von Neumann algebras acting on the code and physical Hilbert spaces, and then we invoke Theorem 1.1 in \cite{HolographicEntropy} to argue that our QECC also satisfies the JLMS formula. We finally show that the relative entropies defined with respect to the infinite-dimensional von Neumann algebras we consider can be expressed as limits of the relative entropies defined with respect to finite-dimensional subalgebras. Thus, another way to see that our QECC satisfies the JLMS formula is to note that our QECC satisfies the JLMS formula with respect to finite-dimensional von Neumann algebras. The JLMS formula for finite-dimensional algebras is studied in \cite{Harlow:2016fse}.

The technical assumptions that connect entanglement wedge reconstruction and the JLMS formula are presented in Theorem 1.1 of \cite{HolographicEntropy}, which we repeat below.

\begin{thm}[Kang-Kolchmeyer \cite{HolographicEntropy}]
\label{thm:maintheorem}
Let $u : \calh_{code}\rightarrow \calh_{phys}$ be an isometry\footnote{This means that $u$ is a norm-preserving map. $u$ need not be a bijection. In general, $u^\dagger u$ is the identity on $\calh_{code}$ and $uu^\dagger$ is a projection on $\calh_{phys}$. } between two Hilbert spaces. Let $M_{code}$ and $M_{phys}$ be von Neumann algebras on $\calh_{code}$ and $\calh_{phys}$ respectively. Let $M^\prime_{code}$ and $M^\prime_{phys}$ respectively be the commutants of $M_{code}$ and $M_{phys}$. Suppose that the set of cyclic and separating vectors with respect to $M_{code}$ is dense in $\calh_{code}$. Also suppose that if $\ket{\Psi} \in \calh_{code}$ is cyclic and separating with respect to $M_{code}$, then $u \ket{\Psi}$ is cyclic and separating with respect to $M_{phys}$. Then the following two statements are equivalent:
	\begin{description}
		\item[ 1. Bulk reconstruction]
		\begin{align} \nonumber
		\begin{split}
		\forall \calo \in M_{code}\ \forall \calo^\prime \in M_{code}^\prime, \quad 
		\exists\tilde{\calo} \in M_{phys}\ \exists \tilde{\calo}^\prime \in M_{phys}^\prime\quad \text{such that}\quad\\
		\forall \ket{\Theta} \in \calh_{code} \quad 
		\begin{cases}
		u \calo \ket{\Theta} =  \tilde{\calo} u \ket{\Theta}, \quad
		&u \calo^\prime \ket{\Theta} =  \tilde{\calo}^\prime u \ket{\Theta}, \\
		u \calo^\dagger \ket{\Theta} =  \tilde{\calo}^\dagger u \ket{\Theta}, \quad
		&u \calo^{\prime \dagger} \ket{\Theta} = \tilde{\calo}^{\prime\dagger} u\ket{\Theta}.
		\end{cases}\quad
		\end{split}
		\end{align}
		
		\item[ 2. Boundary relative entropy equals bulk relative entropy]
		\begin{align}\nonumber
		\begin{split}
		\text{For any $\ket{\Psi}$, $\ket{\Phi} \in \calh_{code}$ with $\ket{\Psi}$ cyclic }&\text{ and separating with respect to $M_{code}$,}\quad\quad\quad\\
		\cals_{\Psi|\Phi}(M_{code})=\cals_{u\Psi|u\Phi}(M_{phys}),\ & \text{and} \  \cals_{\Psi|\Phi}(M_{code}^\prime)= \cals_{u\Psi|u\Phi}(M_{phys}^\prime),\\
		\text{where $\cals_{\Psi|\Phi}(M)$ is the relative entropy.}\quad&
		\end{split}
		\end{align}
	\end{description}
\end{thm}

Tensor networks with a finite number of nodes have been used to construct QECC for finite-dimensional Hilbert spaces, which have yielded physical insights into holography \cite{Harlow:2018fse,Harlow:2016fse}. One such example is the HaPPY code which demonstrates the kinematics of entanglement wedge reconstruction \cite{Pastawski:2015qua}. In particular, the code subspace of the HaPPY code consists of states where the areas of the extremal surfaces are not quantum fluctuating \cite{fixedarea}. Furthermore, some aspects of entanglement wedge reconstruction have also been studied using random tensor networks, where the dimension of each tensor index is finitely large \cite{HaydenQi}. Given the utility of tensor networks for preparing holographic states \cite{beyondtoymodels} and the fact that actual holographic CFTs have infinite-dimensional Hilbert spaces, we expect that infinite-dimensional tensor networks provide additional insights. In particular, an infinite-dimensional tensor network can illustrate the connection between the Reeh-Schlieder theorem and quantum error correction. Furthermore, the modular operator of Tomita-Takesaki theory plays a central role in bulk reconstruction in the continuum limit \cite{faulknerLewkowycz,modulartoolkit}. By linking holographic QECC with infinite-dimensional operator algebras, an infinite-dimensional tensor network might allow one to perform explicit computations relevant to holography that involve the modular operator. An infinite holographic tensor network also has the potential to make boundary locality manifest. In this paper, we demonstrate that the use of tensor networks in quantum error correction can be generalized to the case of infinite-dimensional Hilbert spaces.

In our toy model, the infinite-dimensional code and physical Hilbert spaces are constructed by tensoring together the Hilbert spaces of a countably infinite number of qutrits and then restricting to a countably infinite-dimensional subspace. Finite collections of qutrits are related by a tensor network as represented in Figure \ref{fig:qutritcodes}. Each connected graph defines an isometry from the state of two code qutrits (denoted as black nodes) to the state of four physical qutrits (denoted as white nodes). Our toy model explores how tensor networks with a repeated pattern can be generalized to define a QECC with infinite-dimensional Hilbert spaces. This model does not capture the negatively curved geometry of AdS; however, we believe that our construction can be generalized to encapsulate the holographic setup.
\begin{figure}[H]
	\centering
	\includegraphics[width=0.7\linewidth]{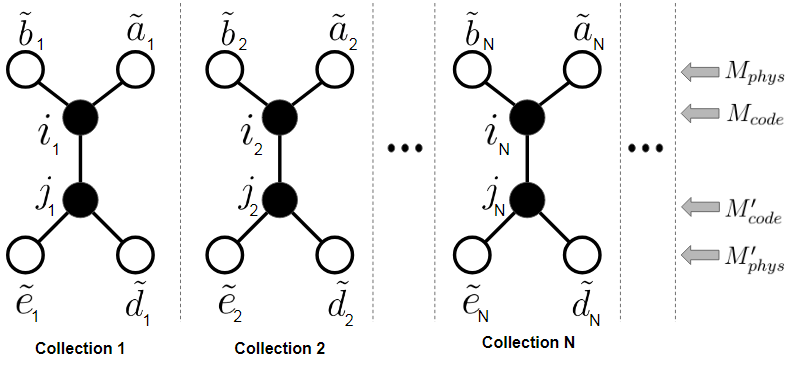}
	\caption{Our setup consists of infinitely many collections of code (black) qutrit pairs which are related to physical (white) qutrits via a tensor network that consists of infinitely many disconnected graphs. The von Neumann algebra $M_{phys}$ acts on the first row of qutrits. The algebras $M_{code}$, $M_{code}^\prime$, and $M_{phys}^\prime$ act on the second, third, and fourth rows of qutrits respectively.}
	\label{fig:qutritcodes}
\end{figure}

A more detailed summary of our construction is given as follows.
\begin{itemize}
\item The code pre-Hilbert space $p\calh_{code}$ is defined to be the Hilbert space of a countably infinite collection of qutrit pairs, where all but finitely many qutrit pairs are in the maximally entangled state $\ket{\lambda} = \frac{1}{\sqrt{3}} (\ket{00} + \ket{11} + \ket{22})$. Each code qutrit pair is represented by two vertically-aligned black nodes in Figure \ref{fig:qutritcodes}. The code Hilbert space $\calh_{code}$ is the completion of $p\calh_{code}$. The physical pre-Hilbert space $p\calh_{phys}$ and physical Hilbert space $\calh_{phys}$ are constructed the same way. Each qutrit pair in the physical Hilbert space is represented by two vertically-aligned white nodes in Figure \ref{fig:qutritcodes}.

\item We construct a bulk-to-boundary isometry from $p\calh_{code}$ to $p\calh_{phys}$ using the tensor network in Figure \ref{fig:qutritcodes}. 
The tensor network is comprised of infinite copies of connected diagrams, where a single connected diagram is represented in Figure \ref{fig:twotofourexample}.
Each trivalent vertex is associated with the rank-four perfect tensor\footnote{A \emph{perfect tensor} is an even-rank tensor that naturally defines an isometric map from up to half of its indices to the remaining indices. For a more detailed discussion of perfect tensors, see \cite{Pastawski:2015qua}.} of the three qutrit code $T_{i\tilde{a}\tilde{b}\tilde{c}}$. Our tensor network maps the states of the black qutrits to the states of the white qutrits. Using these notations, the isometry associated with a connected diagram is explicitly given by
\begin{align*}
\ket{p}_{i} \ket{q}_{j}  \rightarrow \sum_{\tilde{x},\tilde{y},\tilde{z},\tilde{c},\tilde{w}} \sqrt{3} T_{p\tilde{x}\tilde{y}\tilde{c}}  T_{q\tilde{z}\tilde{w}\tilde{c}} \ket{\tilde{x}}_{\tilde{a}}\ket{\tilde{y}}_{\tilde{b}} \ket{\tilde{z}}_{\tilde{d}} \ket{\tilde{w}}_{{\tilde{e}}},
\end{align*}
where the qutrits are labeled as in Figure \ref{fig:qutritcodes}. The indices $p,q,\tilde{x},\tilde{y},\tilde{z},\tilde{w},\tilde{c}$ are all valued in $\{0,1,2\}$. The isometry from $p\calh_{code}$ to $p\calh_{phys}$ may be naturally extended to an isometry that maps $\calh_{code}$ to $\calh_{phys}$.

\item Using the code and physical Hilbert spaces and an isometry relating them, we define von Neumann algebras $M_{code}$ and $M_{phys}$. The $\star$-algebra $A_{code}$ is defined to be the algebra of operators that only act nontrivially on a finite number of qutrits in the top row of black qutrits in Figure \ref{fig:qutritcodes}. The double commutant of $A_{code}$ defines $M_{code}$, the von Neumann algebra acting on the top row of black qutrits. We explicitly show that the commutant $M_{code}^\prime$ is the analogously defined algebra that acts on the bottom row of black qutrits. We also define $M_{phys}$ and $M_{phys}^\prime$ which respectively act on the top and bottom row of white qutrits in Figure \ref{fig:qutritcodes}. To show that $M_{code}$ is a type II$_1$ factor, we define a linear function $T: M_{code} \rightarrow \mathbb{C}$, which is given by
\begin{align*}
T(\calo) := \braket{\lambda \cdots|\calo|\lambda \cdots}, \quad \calo \in M_{code},
\end{align*}
where $\ket{\lambda \cdots} \in p\calh_{code}$ is the state where all black qutrit pairs are in the state $\ket{\lambda}$. We demonstrate that $T(\calo)$ is a trace and invoke Theorem \ref{thm:21theorem} to prove that $M_{code}$ is a type II$_1$ factor. Likewise, $M_{code}^\prime$, $M_{phys}$, and $M_{phys}^\prime$ are also type II$_1$ factors.

\item We determine a map from $M_{code}$ to $M_{phys}$ that explicitly shows how operators that act on black qutrits may be reconstructed as operators that act on white qutrits. First, note that an operator $\calo$ that acts on a black qutrit $i$ in Figure \ref{fig:twotofourexample} may be expressed as an operator $\tilde{\calo}$ that acts on the white qutrits $\tilde{a},\tilde{b}$. The relation between $\calo$ and $\tilde{\calo}$ is given by
\begin{align*}
\tilde{\calo} = \sum_{p,q}\braket{p|\calo|q}_{i} \left[U_{\tilde{a} \tilde{b}}\ket{p}_{\tilde{a}}\bra{q}_{\tilde{a}}  U_{\tilde{a}\tilde{b}}^\dagger \otimes I_{\tilde{d}\tilde{e}}\right],
\end{align*}
where $U_{\tilde{a}\tilde{b}}$ is a unitary matrix that acts only on white qutrits $\tilde{a},\tilde{b}$. By applying the above formula finitely many times, we may construct a map from $A_{code}$ into $M_{phys}$ which we call the \emph{tensor network map}. We then show that there is a natural way to extend the tensor network map to a map from $M_{code}$ into $M_{phys}$. We demonstrate that the image under the tensor network map of an operator $\calo \in M_{code}$ acts on the code subspace in the same way as $\calo$. The same statement holds for the commutant $M_{code}^\prime$. This demonstrates that our QECC satisfies statement 1 of Theorem \ref{thm:maintheorem}.

\item To show that our QECC satisfies the assumptions of Theorem \ref{thm:maintheorem}, we find a dense subset of $p\calh_{code}$ that consists of cyclic and separating vectors with respect to $M_{code}$. For example, a state in $p\calh_{code}$ where each black qutrit pair is in a pure state with maximal Schmidt number (such as $\ket{\lambda}$) is cyclic and separating with respect to $M_{code}$. We also prove that any cyclic and separating state with respect to $M_{code}$ is mapped via the bulk-to-boundary isometry to a cyclic and separating state with respect to $M_{phys}$. Thus, our QECC satisfies all assumptions and statements of Theorem \ref{thm:maintheorem}.
\end{itemize}

An outline of this paper is given as follows. First, we review aspects of infinite-dimensional von Neumann algebras in Section \ref{sec:vNalg} and Tomita-Takesaki theory in Section \ref{sec:tomita}. In particular, we explain why type III$_1$ factors are relevant for quantum field theory. Then, we describe in detail our construction of an infinite-dimensional QECC in Section \ref{sec:isometry}. Tensor networks play an important role in our toy model. In Section \ref{sec:dvna} we define von Neumann algebras $M_{code} \subset \calb(\calh_{code})$ and $M_{phys} \subset \calb(\calh_{phys})$.\footnote{The set of bounded operators on a Hilbert space $\calh$ is denoted by $\calb(\calh)$. See Definition \ref{def:boundedop}.} In Sections \ref{sec:tensornetworkmap} and \ref{sec:tensornetworkproperties}, we show that our example satisfies the properties of bulk reconstruction in Theorem \ref{thm:maintheorem}. In Section \ref{sec:cyclicseparating} we show that cyclic and separating vectors with respect to $M_{code}$ ($M_{phys}$) are dense in $\calh_{code}$ ($\calh_{phys}$). We also show that cyclic and separating vectors with respect to $M_{code}$ are mapped via the bulk-to-boundary isometry $u$ to cyclic and separating vectors with respect to $M_{phys}$. It follows that our tensor network model satisfies both statements in Theorem \ref{thm:maintheorem}. In Section \ref{sec:type21}, we prove that $M_{code}$ and $M_{phys}$ are type II$_1$ factors. In Section \ref{sec:reltomita}, we demonstrate that the relative Tomita operator defined with respect to $M_{code}$ or $M_{phys}$ may be bounded or unbounded, depending on the choice of states. In quantum field theory, the Tomita operators defined with respect to local operator algebras are generically unbounded \cite{Witten:2018zxz}. In section \ref{sec:computerelentropy}, we show that the relative entropy of two cyclic and separating states may be computed by tracing over the entire Hilbert space except the Hilbert space of the first $N$ qutrit pairs, computing the relative entropy of the reduced density matrices with the finite-dimensional relative entropy formula, and taking the limit as $N \rightarrow \infty$.

\section{Infinite-dimensional von Neumann algebras} \label{sec:vNalg}
In this section, we provide background information on operator algebras, including the definitions of type I, II$_1$, II$_\infty$, and III factors, and elucidate their relevance to physics. We also prove theorems that are useful for constructing our infinite-dimensional QECC. First, we review the notion of Hilbert space and bounded operators in Section \ref{sec:hilbert}.  With these notions, we recall basic theorems about Hilbert spaces, operators, and boundedness. Based on those, we explain the operator topologies in Section \ref{sec:optop}. Then we introduce relevant theorems and present definitions of von Neumann algebras in a physics-friendly manner in Section \ref{sec:vNalgdef}. We review the different types of von Neumann algebra factors in Section \ref{sec:vNalgtypes}. This section mainly draws upon \cite{Jones-vNalg}, \cite{Landsman-vNalg}, and \cite{ReedSimon}. 

\subsection{Hilbert Space and Bounded Operators} \label{sec:hilbert}
\begin{defn}
A \emph{Hilbert space} is a complex vector space $\mathcal{H}$ with the inner product 
$$\langle \cdot | \cdot \rangle : \mathcal{H} \times \mathcal{H} \rightarrow \mathbb{C}$$ 
that satisfies the following properties:
\begin{enumerate}
\item The inner product is linear in the second variable, 
\item The inner product satisfies $\overline{\braket{\xi|\eta}} = \braket{\eta|\xi}$, 
\item The inner product is positive definite ($\braket{\xi|\xi} > 0$ for $\ket{\xi} \neq 0$), 
\item The vector space $\calh$ is complete for the norm defined by $||\ket{\xi}|| = \sqrt{\braket{\xi|\xi}}$. 
\end{enumerate}

A Hilbert space is complete when all Cauchy sequences converge. A \emph{pre-Hilbert space} has the same properties as a Hilbert space except that it is not complete.
\end{defn}

\begin{defn}
	A Hilbert space is \emph{separable} when it has an orthonormal basis, or a sequence $\{\ket{e_i}\}$ of unit vectors such that $\braket{e_i|e_j}=0 \quad \forall i \neq j$ and $0$ is the only element of $\mathcal{H}$ orthogonal to all of the $\ket{e_i}$.
	\label{def:orthonormalbasis}
\end{defn}

\begin{defn}
\label{def:boundedop}
	Given two Hilbert spaces $\calh$ and $\calk$, a linear operator $\calo : \mathcal{H} \rightarrow \mathcal{K}$ is \emph{bounded} when $|| \calo \ket{\xi} || \leq K ||\ket{\xi}|| \quad \forall \ket{\xi} \in \mathcal{H}$ for some number $K$. The infimum of all such $K$ is called the \emph{norm} of $\calo$, i.e. $||\calo||$. The set of bounded operators from $\mathcal{H} \rightarrow \mathcal{H}$ is denoted as $\mathcal{B}(\mathcal{H})$.
\end{defn}

\begin{thm}[Uniform Boundedness Principle \cite{Sokal}]
	\label{thm:uniformboundedness}
	Let $\{\calo_n\} \in \mathcal{B}(\mathcal{H})$ be a sequence of operators such that $\lim_{n \rightarrow \infty} \calo_n \ket{\chi}$ converges for every $\ket{\chi} \in \mathcal{H}$. Then, the sequence of norms $\{||\calo_n||\}$ is bounded from above.\footnote{The Uniform Boundedness Principle is true in a more general setting, but we are only interested in the special case given here.}
\end{thm}

\begin{thm}
	\label{thm:convergence}
	If $\{\calo_n\} \in \calb(\calh)$ is a sequence of operators whose norms are bounded from above and $\lim_{n \rightarrow \infty} \calo_n\ket{\psi}$ converges for all $\ket{\psi}$ in a dense subspace of $\calh$, then $\lim_{n \rightarrow \infty} \calo_n\ket{\Psi}$ converges for all $\ket{\Psi} \in \calh$.
\end{thm}
\begin{proof}
	Note that $||(\calo_n-\calo_m)\ket{\Psi}|| \leq (||\calo_n|| + ||\calo_m||)||\ket{\Psi} - \ket{\psi}|| + ||(\calo_n - \calo_m)\ket{\psi}||$. We are given that $ \{\calo_n \ket{\psi}\}$ is a Cauchy sequence. Given $\epsilon > 0$, choose $\ket{\psi}$ such that $(||\calo_n|| + ||\calo_m||)||\ket{\Psi} - \ket{\psi}|| < \frac{\epsilon}{2}$. Then, choose $N$ such that for all $n,m > N$, $ ||(\calo_n - \calo_m)\ket{\psi}||<\frac{\epsilon}{2}$. Hence, $\{\calo_n \ket{\Psi}\}$ is a Cauchy sequence.
\end{proof}

\begin{thm}[Bounded Linear Transformation (BLT) Theorem \cite{ReedSimon}]
	\label{thm:BLT}
	Suppose $\calo$ is a bounded linear transformation from a pre-Hilbert space $p\calh$ to a Hilbert space $\calh$. Then $\calo$ can be uniquely extended to a bounded linear operator (with the same norm) from the completion of $p\calh$ to $\calh$.
\end{thm}

\begin{defn}
	An operator $\calo \in \mathcal{B}(\mathcal{H})$ is
	\begin{itemize}
		\item \emph{self-adjoint} if $\calo^\dagger = \calo$,
		\item a \emph{projection} if $\calo = \calo^\dagger = \calo^2$,
		\item \emph{positive} if $\braket{ \xi|\calo|\xi} \ge 0 \quad \forall \ket{\xi} \in \mathcal{H}$ (thus $\calo_1 \ge \calo_2$ if $\calo_1 - \calo_2$ is positive),
		\item an \emph{isometry} if $\calo^\dagger \calo = 1$,
		\item \emph{unitary} if $\calo^\dagger \calo = \calo \calo^\dagger = 1$,
		\item a \emph{partial isometry} if $\calo^\dagger \calo$ is a projection.
	\end{itemize}
\end{defn}

One can also define an isometry more generally as a norm-preserving map from one Hilbert space to a different Hilbert space. An example is the isometry from $\calh_{code}$ to $\calh_{phys}$ considered in Theorem \ref{thm:maintheorem}.


\begin{defn}
	If $S \subset \mathcal{B}(\mathcal{H})$, then the \emph{commutant} $S^\prime$ is $\{\calo \in \mathcal{B}(\mathcal{H}) : \calo \calp = \calp \calo \quad \forall \calp \in S  \}$.
\end{defn}

\begin{thm}
	\label{thm:strongoperatorconvergence}
	Let $\ket{e_i}, \, i \in \mathbb{N}$ be an orthonormal basis of a Hilbert space $\calh$. Let $\calo \in \calb(\calh)$. Then
	\begin{equation} 
\label{eq:diagionallimit}
	\lim_{n \rightarrow \infty} \, \sum_{i = 1}^n \sum_{j = 1}^n \ket{e_i}\braket{e_i|\calo|e_j}\braket{e_j|\chi} = \calo \ket{\chi} \quad \forall \ket{\chi} \in \calh.
	\end{equation}
\end{thm}

\begin{proof}
	For $n,m \in \mathbb{N}$, define
		\begin{equation} 
	S_{n,m} := \sum_{i = 1}^n \sum_{j = 1}^m \ket{e_i}\braket{e_i|\calo|e_j}\bra{e_j}.
	\end{equation}
	Let $\ket{\chi} \in \calh$. Note that
	\begin{align} 
	\begin{split}
	S_{n,m}\ket{\chi} - \calo\ket{\chi} =   &\left[\sum_{i = 1}^n \ket{e_i}\bra{e_i}\calo - \calo\right]\ket{\chi}   \\
	-&\left[\sum_{i = 1}^n \ket{e_i}\bra{e_i}\calo - \calo\right]\sum_{j = m+1}^\infty \ket{e_j}\braket{e_j|\chi}
	+
	\left[\sum_{j = 1}^m \calo \ket{e_j}\bra{e_j} - \calo\right]\ket{\chi}.
	\end{split}
	\end{align}
We will evaluate the norm of the above equation and use the triangle inequality on the right hand side. We need the inequality
\begin{align} 
\begin{split}
	|| &\left[\sum_{i = 1}^n \ket{e_i}\bra{e_i}\calo - \calo\right]\sum_{j = m+1}^\infty\ket{e_j}\braket{e_j|\chi}|| \leq || \sum_{i = 1}^n \ket{e_i}\bra{e_i}\calo - \calo || \cdot || \sum_{j = m+1}^\infty\ket{e_j}\braket{e_j|\chi}||
	\\
	 &\leq K ||\sum_{j=m+1}^\infty\ket{e_j}\braket{e_j|\chi}||,
	\end{split}
	\end{align}
	where $K > 0$ is some constant. This inequality follows from the fact that the limit
	\begin{equation}
	\lim_{n \rightarrow \infty}\left[\sum_{i = 1}^n \ket{e_i}\bra{e_i}\calo - \calo\right] \ket{\psi}\end{equation}
	converges for all $\ket{\psi} \in \calh$, which implies that the set $$\{||\sum_{i = 1}^n \ket{e_i}\bra{e_i}\calo - \calo|| : n \in \mathbb{N} \}$$ is bounded (see Theorem \ref{thm:uniformboundedness}). Thus,
\begin{align}
\begin{split}
||S_{n,m}\ket{\chi} - \calo\ket{\chi}|| \leq   &||\left[\sum_{i = 1}^n \ket{e_i}\bra{e_i}\calo - \calo\right]\ket{\chi}|| 
	\\
	&+ K ||\sum_{j = m+1}^\infty\ket{e_j}\braket{e_j|\chi}||
	+ ||\left[\sum_{j = 1}^m \calo \ket{e_j}\bra{e_j} - \calo\right]\ket{\chi}||.
\end{split}
\end{align}
Given any $\epsilon > 0$, there exists an $M \in \mathbb{N}$ such that for $m > M$, 
\begin{equation}
K ||\sum_{j = m+1}^\infty\ket{e_j}\braket{e_j|\chi}||<\frac{\epsilon}{3}, \quad ||\left[\sum_{j = 1}^m \calo \ket{e_j}\bra{e_j} - \calo\right]\ket{\chi}||<\frac{\epsilon}{3}.
\end{equation}
There also exists an $N \in \mathbb{N}$ such that for $n > N$, 
\begin{equation}
||\left[\sum_{i = 1}^n \ket{e_i}\bra{e_i}\calo - \calo\right]\ket{\chi}||<\frac{\epsilon}{3}.
\end{equation}
Thus, there exist $N,M \in \mathbb{N}$ such that $||S_{n,m}\ket{\chi} - \calo\ket{\chi}||<\epsilon$ for $n > N$ and $m > M$. Hence,
\begin{equation}
\lim_{n \rightarrow \infty} \sum_{i = 1}^n \sum_{j = 1}^n \ket{e_i}\braket{e_i|\calo|e_j}\braket{e_j|\chi} = \lim_{n \rightarrow \infty} S_{n,n}\ket{\chi} = \calo\ket{\chi} \quad \forall \ket{\chi} \in \calh.
\end{equation}
\end{proof}

\begin{rem}
Naively, the equation \eqref{eq:diagionallimit} in Theorem \ref{thm:strongoperatorconvergence} can be thought of as a trivial consequence of the statement that for all $\ket{\chi}$ in $\calh$,
\begin{align*}
\begin{split}
\lim_{n \rightarrow \infty} \lim_{m \rightarrow \infty} \, \sum_{i = 1}^n \sum_{j = 1}^m \ket{e_i}\braket{e_i|\calo|e_j}\braket{e_j|\chi} &= \lim_{m \rightarrow \infty} \lim_{n \rightarrow \infty} \, \sum_{i = 1}^n \sum_{j = 1}^m \ket{e_i}\braket{e_i|\calo|e_j}\braket{e_j|\chi} = \calo \ket{\chi} .
\end{split}
\end{align*}
However we note that Theorem \ref{thm:strongoperatorconvergence} is nontrivial; we demonstrate this by the following counter example where the statement holds where the equation \eqref{eq:diagionallimit} does not hold. Consider the double-sequence $a_{n,m} \in \mathbb{R}$, indexed by $n,m \in \mathbb{N}$, which is defined as
\begin{equation}
a_{n,m} := \frac{1}{\frac{m}{n} + \frac{n}{m}}.
\end{equation}
One can check that
\begin{equation}
\lim_{n \rightarrow \infty} \lim_{m \rightarrow \infty} a_{n,m} = \lim_{m \rightarrow \infty} \lim_{n \rightarrow \infty} a_{n,m} = 0.
\end{equation}
However, we get a nonzero limit of $a_{n,n}$ such that
\begin{equation}
\lim_{n \rightarrow \infty} a_{n,n} = \frac{1}{2}.
\end{equation}
This demonstrates that the Theorem \ref{thm:strongoperatorconvergence} is not a simple consequence of the definition of a limit. Our proof of Theorem \ref{thm:strongoperatorconvergence} makes use of Theorem \ref{thm:uniformboundedness}, which is demonstrated above in its proof.
\end{rem}

\subsection{Topologies on $\calb(\calh)$}\label{sec:optop}

A topology on $\mathcal{B}(\mathcal{H})$ is a family of subsets of $\calb(\calh)$ that are defined to be open. This family must contain both the empty set $\emptyset$ and $\mathcal{B}(\mathcal{H})$ itself. Furthermore, this family must be closed under finite intersections and arbitrary unions. There are various notions of open sets in $\mathcal{B}(\mathcal{H})$; we list their definitions below, closely following \cite{Jones-vNalg}. In this section $\calo$ denotes an operator in $\calb(\calh)$ and $\ket{\xi_i}, \ket{\eta_i}$ denote states in $\calh$.


\begin{defn}
	The \emph{norm (or uniform) topology} is induced by the operator norm $||\calo||$. It is the smallest topology that contains the following basic neighborhoods:
	\begin{equation} \nonumber
	\mathcal{N}(\calo,\epsilon) = \{\calp \in \calb(\calh) : ||\calp - \calo|| < \epsilon  \} .
	\end{equation}
\end{defn}
\begin{defn}
The \emph{strong operator topology} is the smallest topology that contains the following basic neighborhoods:
\begin{equation} \nonumber
\mathcal{N}(\calo,\ket{\xi_1},\ket{\xi_2},\ldots,\ket{\xi_n},\epsilon) = \{\calp \in \calb(\calh) : ||(\calp - \calo)\ket{\xi_i}|| < \epsilon \quad \forall i \in \{ 1,2,\cdots,n\} \} .
\end{equation}
\end{defn}
A sequence of bounded operators $\{\calo_n\}$ converges strongly if and only if $\lim_{n \rightarrow \infty} \calo_n \ket{\psi}$ converges for all $\ket{\psi} \in \calh$. Note that the hermitian conjugates $\calo_n^\dagger$ need not converge strongly. We will sometimes use $\slim$ to denote a strong limit.
\begin{defn}
The \emph{weak operator topology} is the smallest topology that contains the following basic neighborhoods:
\begin{equation} \nonumber
\mathcal{N}(\calo,\ket{\xi_1},\ldots,\ket{\xi_n},\ket{\eta_1},\ldots,\ket{\eta_n},\epsilon) = \{\calp \in \calb(\calh): |\braket{\eta_i|(\calp - \calo)|\xi_i}| < \epsilon \quad \forall i \in \{1,2,\cdots,n\} \} .
	\label{def:weakoperatortopology}
\end{equation}
\end{defn}
A sequence of bounded operators $\{\calo_n\}$ converges weakly if and only if $\lim_{n \rightarrow \infty} \braket{\chi|\calo_n|\psi}$ converges for all $\ket{\chi},\ket{\psi} \in \calh$. We will sometimes use $\wlim$ to denote a weak limit.

\begin{defn}
	\label{def:ultraweakoperatortopology}
The \emph{ultraweak operator topology} is the smallest topology that contains the following basic neighborhoods:
\begin{equation} \nonumber 
\mathcal{N}(\calo,\{\ket{\xi_i}\},\{\ket{\eta_i}\},\epsilon) = \{\calp \in \calb(\calh) : \sum_{i = 1}^\infty|\braket{\eta_i|(\calp-\calo)|\xi_i}| < \epsilon \} ,
\end{equation} 
where the sequences $\{\ket{\xi_i}\}$ and $\{\ket{\eta_i}\}$ satisfy 
\begin{equation} \nonumber
\sum_{i = 1}^\infty (||\ket{\xi_i}||^2+||\ket{\eta_i}||^2) < \infty. 
\end{equation}
\end{defn}

Given topologies A and B, we say that topology A is stronger than topology B when every open set in topology B is also open in topology A. The relations between the various operator topologies are given in Figure \ref{fig:topologies}.

\begin{figure}[H]
\centering
\begin{tikzcd}[column sep=tiny, row sep=normal]
 & \text{Norm operator topology}\arrow[rightarrow]{ld}\arrow[rightarrow]{rd} &\\
\text{Strong operator topology}\arrow[rightarrow]{rd} & & \text{Ultraweak operator topology}\arrow[rightarrow]{ld}\\
 & \text{Weak operator topology}
\end{tikzcd}
\caption{As shown in \cite{Takesaki}, the norm operator topology is stronger than the strong operator topology and the ultraweak operator topology, which are both stronger than the weak operator topology.}
\label{fig:topologies}
\end{figure}
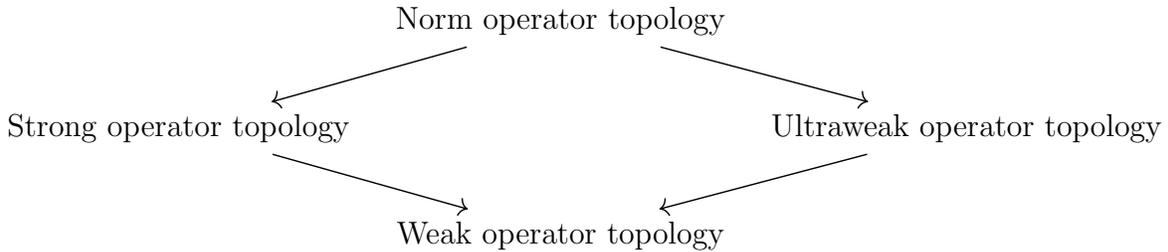

\subsection{Definition of von Neumann algebras} 
\label{sec:vNalgdef}
In this section, we define von Neumann algebras, factors, and hyperfinite von Neumann algebras.

\begin{defn}
	A \emph{$\star$-algebra} is an algebra of operators that is closed under hermitian conjugation.
\end{defn}

\begin{thm}[\cite{Jones-vNalg}, page 12]
	Let M be a $\star$-subalgebra of $\mathcal{B}(\mathcal{H})$ that contains the identity operator. Then $M^{\prime \prime} = \overline{M}$, where closure\footnote{A set is closed if its complement is open. The closure of a set $S$, denoted $\bar{S}$, is the smallest closed set that contains $S$.} is taken in the strong operator topology.
\end{thm}

\begin{thm}[\cite{Jones-vNalg}, page 12]
\label{thm:doublecommutant}
	If $M$ is a $\star$-subalgebra of $\mathcal{B}(\mathcal{H})$ that contains the identity operator, then the following statements are equivalent:
	\begin{itemize}
		\item $M = M^{\prime \prime}$,
		\item $M$ is closed in the strong operator topology,
		\item $M$ is closed in the weak operator topology.
	\end{itemize}
\end{thm}
\begin{defn}
	A \emph{von Neumann algebra} is an algebra that satisfies the statements in Theorem \ref{thm:doublecommutant}.
\end{defn}
Given a $\star$-subalgebra of $\mathcal{B}(\mathcal{H})$ containing the identity, we can generate a von Neumann algebra by taking either the double commutant or the closure in the strong or weak topology. 

\begin{defn}
A \emph{factor} is a von Neumann algebra $M$ with trivial center. That is, $$M \cap M^\prime = \{ \lambda  I : \lambda \in \mathbb{C} \},$$
where $I$ denotes the identity operator.
\end{defn}

\begin{defn}
	A von Neumann algebra $M$ is $\emph{hyperfinite}$ if $M = (\cup_n M_n)^{\prime \prime}$ for a sequence $\{M_n\}$ of finite-dimensional von Neumann subalgebras of $M$ that satisfies $M_n \subset M_{n+1} \quad \forall n \in \mathbb{N}$.
\end{defn}

Note that the union of finitely many closed sets is also closed. However, the union of infinitely many closed sets need not be closed. In section \ref{sec:dvna}, we define a hyperfinite von Neumann algebra by taking the closure of an infinite union of finite-dimensional von Neumann algebras. The closure introduces additional operators into the algebra.


\begin{defn}
	If $M$ is a von Neumann algebra, a non-zero projection $p \in M$ is called \emph{minimal} if, for any other projection $q$, $q \leq p \implies (q = 0 \text{ or } q = p)$.
\end{defn}

\begin{defn}
	Let $A$ be a $\star$-algebra that contains the identity operator $I$. Let $T : A \rightarrow \mathbb{C}$ be a linear function on $A$. The map $T$ is
	\begin{itemize}
		\item	\emph{positive} if $T(\calo^\dagger \calo) \ge 0$ and $T(\calo^\dagger) = T(\calo)^* \quad \forall \calo \in A$, 
		\item \emph{normalized} if $T(I)$ = 1,
		\item 	a \emph{state} if $T$ is positive and normalized, 
		\item	\emph{faithful} if $T(\calo^\dagger \calo) = 0 \implies \calo = 0 \quad \forall \calo \in A$, 
		\item 	\emph{tracial} (or a \emph{trace}) if $T(\calo_1 \calo_2) = T(\calo_2 \calo_1) \quad \forall \calo_1,\calo_2 \in A$.
	\end{itemize}
\end{defn}

Given any normalized Hilbert space vector $\ket{\Psi}$ and a von Neumann algebra $M \subset \calb(\calh)$, one can naturally define an associated state $T_\Psi : M \rightarrow \mathbb{C}$ as
\[ T_\Psi(\calo) := \braket{\Psi|\calo|\Psi} \quad \forall \calo \in M. \]
For this reason, the term ``state'' is often used to refer to both Hilbert space vectors and positive, normalized linear functions of a von Neumann algebra.

\subsection{Classification of von Neumann algebras}
\label{sec:vNalgtypes}
In this section, we review the classification of von Neumann algebra factors in a manner to have a direct consequence in physics. We first review type I factors, which are the only factors relevant for finite-dimensional Hilbert spaces. We then review type II factors. We explicitly construct type II$_1$ factors from our tensor network model in Section \ref{sec:dvna}. We finally review type III factors. We explain why among type III factors we only expect type III$_1$ factors to arise as algebras in local quantum field theories. 

\subsubsection{Type I factors}

\begin{defn}
	A factor with a minimal projection is called a \emph{type I factor}.
\end{defn}

\begin{defn}
	A type I factor that is isomorphic\footnote{We say that the von Neumann algebras $M_1$ and $M_2$, which may act on different Hilbert spaces, are \emph{isomorphic} when there exists a bijection between $M_1$ and $M_2$ that preserves linear combinations, products, and adjoints. We refer the reader to section III.2.1 of \cite{Haag} for more details.} to the algebra of bounded operators on a Hilbert space of dimension $n$ is a \emph{type I$_n$ factor}.
\end{defn}

\begin{defn}
	A type I factor that is isomorphic to the algebra of bounded operators on an infinite-dimensional Hilbert space is a \emph{type I$_\infty$ factor}.
\end{defn}

\subsubsection{Type II factors}

\begin{defn}
	\label{def:type21}
	\emph{A type II$_1$ factor} is an infinite-dimensional factor $M$ on $\mathcal{H}$ that admits a non-zero linear function $\text{tr}: M \rightarrow \mathbb{C}$ satisfying the following properties:
	\begin{itemize}
		\item tr$(\calo_1 \calo_2)=$ tr$(\calo_2 \calo_1) \quad \forall \calo_1,\calo_2 \in M$,
		\item tr$(\calo^\dagger \calo)\ge0 \quad\forall \calo \in M$,
		\item tr is ultraweakly continuous.
	\end{itemize}
\end{defn}

\begin{thm}[\cite{Jones-vNalg}, page 39]
	Let $M$ be a von Neumann algebra with a positive ultra-weakly continuous faithful normalized trace $\text{tr}$. Then $M$ is a type II$_1$ factor if and only if $\trace = \text{tr}$ for all ultraweakly continuous normalized traces $\trace$.
	\label{thm:21theorem}
\end{thm}

\begin{thm}[\cite{Jones-vNalg}, page 109]
	Up to isomorphisms, there is a unique hyperfinite type II$_1$ factor.
\end{thm}

\begin{defn}[\cite{Jones-vNalg}, page 57]
	A type II$_\infty$ factor is a factor of the form $M \otimes \calb(\calh)$ with $M$ a type II$_1$ factor and $\dim \calh = \infty$.
\end{defn}


\subsubsection{Type III factors}

In order to define the von Neumann algebra of type III factor, we first recall from \cite{Jones-vNalg} the definition of the invariant $S(M)$ using the modular operator, which is presented in section \ref{sec:tomita}.

\begin{defn}
	\label{def:Sintersection}
	If $M$ is a von Neumann algebra, the invariant $S(M)$ is the intersection over all faithful normal states $\phi$ of the spectra of their corresponding modular operators $\Delta_{\phi}$.
\end{defn}

Note that each cyclic and separating vector in the Hilbert space defines a faithful normal state. Thus, for every cyclic and separating vector $\ket{\Psi}$, $S(M)$ is a subset of the spectrum of the modular operator $\Delta_\Psi$. With the intersection $S(M)$, we can define the type III factor.

\begin{defn}
	\label{def:typeIII}
	A factor $M$ is of type III if and only if $0\in S(M)$.
\end{defn}

When $0\in S(M)$, every modular operator $\Delta_{\Psi}$ is not a bijection of $D(\Delta_{\Psi})$ onto $\calh$.\footnote{This is a direct consequence of the definitions of the spectrum and the resolvent set. We let $D(\calo)$ denote the domain of operator $\calo$. See section 2.2 of \cite{HolographicEntropy} for more information.
		
	\begin{defn}
		The \emph{spectrum} of $\calo \in \calb(\calh)$ is defined as
		\begin{equation*}
		\sigma(\calo) := \{ \lambda \in \mathbb{C} : \calo - \lambda I \text{ is not invertible}\},
		\end{equation*}
		where $I$ denotes the identity operator.
	\end{defn}
	\begin{defn}
		Let $\calo$ be a closed operator on a Hilbert space $\calh$. $\lambda \in \mathbb{C}$ is in the \emph{resolvent set} of $\calo$ if $\lambda I - \calo$ is a bijection of $D(\calo)$ onto $\calh$. The \emph{spectrum} of $\calo$, denoted $\sigma(\calo)$, is defined to be the set of all complex numbers that are not in the resolvent set of $\calo$.
	\end{defn}
} It follows that the inverse of every modular operator is not defined on the entire Hilbert space. This is exactly desired for a local quantum field theory because the inverse of a modular operator is the modular operator defined with respect to the commutant:
\begin{align}
\Delta_{\Psi}^{-1}=\Delta_{\Psi}^\prime.
\end{align}
As shown in \cite{Witten:2018zxz}, $\Delta_{\Psi}^\prime$ should not be bounded and thus should not be defined on the entire Hilbert space. If $0 \notin S(M)$, then there exists a state whose modular operator defined with respect to $M^\prime$ is bounded. Hence we expect the condition $0 \in S(M)$ to be satisfied by the algebras arising from a physical local quantum field theory.

\begin{defn}
	\label{def:typeIIIlambda}
	A factor $M$ is called type III$_\lambda$ for $0\leq\lambda\leq 1$ if
	\begin{align}
	\lambda =0  \implies&	\quad S(M)=\{0\}\cup\{1\} ,\\
	0<\lambda <1 \implies&	\quad S(M)=\{0\}\cup\{\lambda^n: \, n\in\mathbb{Z}\} ,\\
	\lambda =1 \implies&	\quad S(M)=\{0\}\cup\mathbb{R}^+ .
	\end{align}
\end{defn}

As explained in \cite{Witten:2018zxz}, we expect a local quantum field theory to have a continuous spectrum of the modular operator $\Delta_{\Psi}$. Thus we see that the von Neumann algebra of type III$_1$ factor is the only factor that is relevant to physics among all possible type III factors.

We can also use $S(M)$ to characterize factors of types I or II. For such factors, $S(M)$ is given by the following theorem.
\begin{thm}[\cite{Summers}]
	\label{thm:typeIandIIinvariant}
	Let $M$ be a type I or type II factor on a separable Hilbert space. Let $S(M)$ be the invariant given in Definition \ref{def:Sintersection}. Then $S(M) = \{0,1\}$ if $M$ is of type I$_\infty$ or II$_\infty$ and $S(M) = \{1\}$ otherwise.
\end{thm}


\section{Relative Entropy from Tomita-Takesaki theory}
\label{sec:tomita}

In this section, we review aspects of Tomita-Takesaki theory that are relevant to Theorem \ref{thm:maintheorem}. In particular, we need these definitions to show that our QECC satisfies the assumptions of Theorem \ref{thm:maintheorem}.
For a more thorough review, see Section 3 of \cite{HolographicEntropy} as well as \cite{Witten:2018zxz}.

\begin{defn}
\label{def:cyc}
A vector $\ket{\Psi} \in \calh$ is {\em cyclic} with respect to a von Neumann algebra $M$ when the set of vectors $\calo\ket{\Psi}$ for $\calo \in M$ is dense in $\calh$. 
\end{defn}
\begin{defn}
\label{def:sep}
A vector $\ket{\Psi} \in \calh$ is {\em separating} with respect to a von Neumann algebra $M$ when zero is the only operator in $M$ that annihilates $\ket{\Psi}$. That is, $\calo\ket{\Psi} = 0 \implies \calo = 0$ for $\calo \in M$. 
\end{defn}

\begin{defn}
	Let $\ket{\Psi},\ket{\Phi} \in \calh$ and let $M$ be a von Neumann algebra. 
	The {\em relative Tomita operator} is the operator $S_{\Psi | \Phi}$ that acts as
	\begin{equation*}
	S_{\Psi|\Phi} \ket{x} := \ket{y}
	\end{equation*}
	for any sequence $\{\calo_n\} \in M$ such that the limits $\ket{x} = \lim_{n \rightarrow \infty} \calo_n \ket{\Psi}$ and $\ket{y} = \lim_{n \rightarrow \infty} \calo_n^\dagger \ket{\Phi}$ both exist.
\end{defn}

For this definition to make sense, $\ket{\Psi}$ must be cyclic and separating with respect to $M$.

\begin{defn}
	Let $S_{\Psi|\Phi}$ be a relative Tomita operator.
	 The {\em relative modular operator} is $$\Delta_{\Psi|\Phi} := S_{\Psi|\Phi}^\dagger S_{\Psi|\Phi}.$$
\end{defn}

\begin{defn}[\cite{Araki}] \label{def:relent}
	Let $\ket{\Psi},\ket{\Phi} \in \calh$ and let $\ket{\Psi}$ be cyclic and separating with respect to a von Neumann algebra $M$. Let $\Delta_{\Psi|\Phi}$ be the relative modular operator associated with $\ket{\Psi},\ket{\Phi},$ and $M$. The {\em relative entropy} with respect to $M$ of $\ket{\Psi}$ and $\ket{\Phi}$ is
	\begin{equation*} 
	\cals_{\Psi|\Phi}(M) = - \braket{\Psi|\log \Delta_{\Psi | \Phi}|\Psi}. 
	\end{equation*} 
\end{defn}

\begin{defn}
	Let $M$ be a von Neumann algebra on $\calh$ and $\ket{\Psi}$ be a cyclic and separating vector for $M$. The \emph{Tomita operator} $S_\Psi$ is
	$$ S_\Psi := S_{\Psi | \Psi},$$ where $S_{\Psi | \Psi}$ is the relative modular operator defined with respect to $M$. The \emph{modular operator} $\Delta_\Psi = S_\Psi^\dagger S_\Psi$ and the \emph{antiunitary operator} $J_\Psi$ are the operators that appear in the polar decomposition of $S_\Psi$ such that
	$$S_\Psi = J_\Psi \Delta_{\Psi}^{1/2}.$$
\end{defn}

\section{The isometry between two infinite-dimensional Hilbert spaces}
\label{sec:isometry}
In this section, we show how a tensor network with infinitely many nodes can be used to define an isometry (i.e. a norm preserving map) from one infinite-dimensional Hilbert space to another. The isometry will be denoted by $u : \calh_{code} \rightarrow \calh_{phys}$.  We first review some preliminary facts about the three qutrit code.

\subsection{The three-qutrit code and a finite tensor network}
The three-qutrit code is an example of a QECC. A code qutrit is isometrically mapped to a Hilbert space of three physical qutrits. The map is given by
\begin{align}
\begin{cases}
\ket{0} &\rightarrow \frac{1}{\sqrt{3}} (\ket{\tilde{0}\tilde{0}\tilde{0}} + \ket{\tilde{1}\tilde{1}\tilde{1}} + \ket{\tilde{2}\tilde{2}\tilde{2}}), \\
\ket{1} &\rightarrow \frac{1}{\sqrt{3}} (\ket{\tilde{0}\tilde{1}\tilde{2}} + \ket{\tilde{1}\tilde{2}\tilde{0}} + \ket{\tilde{2}\tilde{0}\tilde{1}}), \\ \ket{2} &\rightarrow \frac{1}{\sqrt{3}} (\ket{\tilde{0}\tilde{2}\tilde{1}} + \ket{\tilde{1}\tilde{0}\tilde{2}} + \ket{\tilde{2}\tilde{1}\tilde{0}}).
\end{cases}
\end{align}
We can write this more succinctly as
\begin{equation} 
\ket{i} \rightarrow  \sum_{\tilde{a},\tilde{b},\tilde{c}}T_{i\tilde{a}\tilde{b}\tilde{c}} \ket{\tilde{a}\tilde{b}\tilde{c}},
\end{equation}
where $i$ denotes an input leg and $\tilde{a}, \tilde{b}, \tilde{c}$ denote output legs. We can apply successive isometries to create an isometry from two code qutrits to four physical qutrits. We illustrate this with a tensor network, represented in Figure \ref{fig:twotofourexample}.
\begin{figure}[H]
\begin{center}
\includegraphics[width=0.3\linewidth]{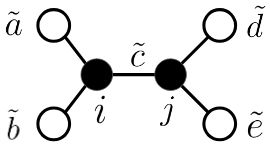}
\caption{The black code subspace qutrits are mapped to the white physical qutrits via an isometry defined by this tensor network.}
\label{fig:twotofourexample}
\end{center}
\end{figure}
The isometry corresponding to Figure \ref{fig:twotofourexample} is given by
\begin{equation} 
\ket{p}_{i} \ket{q}_{j}  \rightarrow \sum_{\tilde{x},\tilde{y},\tilde{z},\tilde{c},\tilde{w}} \sqrt{3} T_{p\tilde{x}\tilde{y}\tilde{c}}  T_{q\tilde{z}\tilde{w}\tilde{c}} \ket{\tilde{x}}_{\tilde{a}}\ket{\tilde{y}}_{\tilde{b}} \ket{\tilde{z}}_{\tilde{d}} \ket{\tilde{w}}_{{\tilde{e}}}. \label{eq:onefactorisometry}
\end{equation}
Throughout this paper, we use subscripts to associate qutrits with specific nodes in Figures \ref{fig:twotofourexample} or \ref{fig:qutritcodes}, and tildes are used to denote qutrits in the physical Hilbert space.

Let $U$ be a unitary operator that acts on a two-qutrit state as
\begin{equation}
\label{eq:unitary}
\begin{array}{ccc}
U\ket{00} = \ket{00} & U\ket{11} = \ket{20} & U\ket{22} = \ket{10} \\ 
U\ket{01} = \ket{11} & U\ket{12} = \ket{01} & U\ket{20} = \ket{21} \\ 
U\ket{02} = \ket{22} & U\ket{10} = \ket{12} & U\ket{21} = \ket{02},
\end{array} 
\end{equation}
and define
\begin{equation} 
\label{eq:lambda}
\ket{\lambda} := \frac{1}{\sqrt{3}}\left[\ket{00}+\ket{11}+\ket{22}\right]. \end{equation}
Let $\ket{\psi}_{ij}$ be a vector in the Hilbert space of the black qutrits $i,j$ in Figure \ref{fig:twotofourexample}, and let $\ket{\tilde{\psi}}_{\tilde{a}\tilde{b}\tilde{d}\tilde{e}}$ be its image under the isometry in equation \eqref{eq:onefactorisometry}. Let $U_{\tilde{a}\tilde{b}}$ ($U_{\tilde{d}\tilde{e}}$) be the unitary operator in equation \eqref{eq:unitary} that acts on qutrits $\tilde{a},\tilde{b}$ ($\tilde{d},\tilde{e}$). One may compute that
\begin{equation} \label{eq:codetophysunitary} U^\dagger_{\tilde{a}\tilde{b}} U^\dagger_{\tilde{d}\tilde{e}}\ket{\tilde{\psi}}_{\tilde{a}\tilde{b}\tilde{d}\tilde{e}} = \ket{\psi}_{\tilde{a} \tilde{d}} \ket{\lambda}_{\tilde{b} \tilde{e}}, \end{equation}
where $\ket{\psi}_{\tilde{a}\tilde{d}}$ is the same state as $\ket{\psi}_{ij}$, except on the white qutrits $\tilde{a}, \tilde{d}$.  That is, starting with the state $\ket{\tilde{\psi}}_{\tilde{a}\tilde{b}\tilde{d}\tilde{e}}$ on the white qutrits, one can apply separate unitary transformations on white qutrits $\tilde{a},\tilde{b}$ and $\tilde{d},\tilde{e}$ to recover $\ket{\psi}_{ij}$ on white qutrits $\tilde{a},\tilde{d}$ and the maximally entangled state $\ket{\lambda}$ on qutrits $\tilde{b},\tilde{e}$.


Given an operator $\calo$ that acts on qutrit $i$ in Figure \ref{fig:twotofourexample}, we may define an operator $\tilde{\calo}$ that acts on the adjacent white qutrits $\tilde{a},\tilde{b}$ as follows:
\begin{equation} \tilde{\calo} := \sum_{p,q}\braket{p|\calo|q}_{i} \left[U_{\tilde{a} \tilde{b}}\ket{p}_{\tilde{a}}\bra{q}_{\tilde{a}}  U_{\tilde{a}\tilde{b}}^\dagger \otimes I_{\tilde{d}\tilde{e}}\right]. 
\label{eq:onefactormap}
\end{equation}

We say that $\calo$, which acts on the code Hilbert space, is reconstructed as $\tilde{\calo}$, which acts on the physical Hilbert space. 

\subsection{The code and physical Hilbert spaces}

Our general setup is depicted in Figure \ref{fig:qutritcodes}. In our construction of an infinite-dimensional QECC, the code and physical Hilbert spaces, $\calh_{code}$ and $\calh_{phys}$, are each defined as the completions of pre-Hilbert spaces, $p\calh_{code}$ and $p\calh_{phys}$. As Figure \ref{fig:qutritcodes} shows, we may intuitively think of either the code or physical pre-Hilbert space as an infinite tensor product of two black qutrits or four white qutrits. From now on, whenever we say \emph{collection} we are referring to the qutrits in a single connected diagram in Figure \ref{fig:qutritcodes}. Within each collection we will label the individual qutrits as shown in Figure \ref{fig:twotofourexample}.

The pre-Hilbert space $p\calh_{code}$ is defined to include states of black qutrits where all but finitely many pairs of black qutrits are in the state $\ket{\lambda}$, defined in equation \eqref{eq:lambda}, which we sometimes also refer to as the \emph{code reference state}. Any vector in $p\calh_{code}$ is a finite linear combination of vectors in an overcomplete basis, where each basis vector may be written as
\begin{equation} \ket{M,p_1,p_2,\cdots,p_M,q_1,q_2,\cdots,q_M} := \left[\ket{p_1}_{i_1} \ket{q_1}_{j_1}\right] \otimes \left[\ket{p_2}_{i_2} \ket{q_2}_{j_2}\right] \otimes \cdots \otimes \left[\ket{p_M}_{i_M} \ket{q_M}_{j_M}\right] \otimes \ket{\lambda} \cdots \label{eq:codeprehbasis}\end{equation}
where each $p_k$ or $q_k$ index (for $k \in \{1,2,\ldots,M\}$) is valued in $\{0,1,2\}$ and specifies an orthonormal basis vector of a black qutrit. The index $M$ can be any natural number. The qutrits in each collection are contained in square brackets. To shorten notation, we will refer to the above basis vector as $\ket{M,\{p,q\}}$. The $\otimes \ket{\lambda} \cdots$ means that all the black qutrit pairs in the $(M+1)$th collection and beyond are in the reference state $\ket{\lambda}$. 
Note that these basis vectors are not all linearly independent.

Given two basis vectors $\ket{M_1,\{p,q\}_1}$ and $\ket{M_2,\{p,q\}_2}$, their inner product is calculated by ignoring all collections beyond the collection $\max{(M_1,M_2)}$ and then taking the usual inner product on the remaining $9^{\max{(M_1,M_2)}}$-dimensional Hilbert space. Note that the basis vectors $\ket{M,\{p,q\}}$ are not all mutually orthogonal, but they are all normalized. With an inner product, we can define Cauchy sequences. The Hilbert space $\calh_{code}$ is defined as the closure of $p\calh_{code}$ so that all Cauchy sequences in $\calh_{code}$ converge. We start from $p\calh_{code}$ and include all Cauchy sequences to define $\calh_{code}$. If the difference of two Cauchy sequences converges to zero, then we identify the two Cauchy sequences for the purposes of defining $\calh_{code}$.

The physical pre-Hilbert and Hilbert spaces are defined in a completely analogous way. Each collection consists of four white qutrits. The \emph{physical reference state} for four white qutrits is given by $\ket{\lambda \lambda} := \ket{\lambda}_{\tilde{a}\tilde{d}} \ket{\lambda}_{\tilde{b}\tilde{e}}$ where we are referring to Figure \ref{fig:twotofourexample} to label the qutrits. We choose this reference state for the white qutrits because it is the image of $\ket{\lambda}_{ij}$ under the isometry given by equation \eqref{eq:onefactorisometry}.

\subsection{The tensor network of isometries}

The bulk-to-boundary isometry $u$ is given by a linear norm preserving map $u : \calh_{code} \rightarrow \calh_{phys}$. First, we define its action on $p\calh_{code}$ and then use Theorem \ref{thm:BLT} to extend its domain to $\calh_{code}$. Each vector in $p\calh_{code}$ is mapped to a vector in $p\calh_{phys}$. The isometry $u$ acts on the basis vector $\ket{M,\{p,q\}}$ by applying the isometry given in equation \eqref{eq:onefactorisometry} to each collection separately. The state of each black qutrit pair is mapped to a state of four white qutrits. The code reference state is mapped to the physical reference state.
Because the map $u$ is linear and norm-preserving, a Cauchy sequence in $p\calh_{code}$ is mapped to a Cauchy sequence in $p\calh_{phys}$. Thus, we can define $u$ on all of $\calh_{code}$.

\section{Defining von Neumann algebras}
\label{sec:dvna}
Now that we have defined $\calh_{code}$ and $\calh_{phys}$ and the isometry $u : \calh_{code} \rightarrow \calh_{phys}$, we want to define von Neumann algebras on these Hilbert spaces.

\subsection{Definition of $M_{code}$}
\label{sec:mcode}
We now define $M_{code} \subset \calb(\calh_{code})$. First, we define a $\star$-algebra called $A_{code}$ which acts on $p\calh_{code}$. Referring to Figure \ref{fig:twotofourexample} for qutrit labels, every operator $a \in A_{code}$ may be written as
\begin{align}
\begin{split}
\label{eq:opinacode} 
a^{(N)} = \sum_{p_1 p_2 \cdots p_N q_1 q_2 \cdots q_N} & a_{p_1 p_2 \cdots p_N q_1 q_2 \cdots q_N} \times 
\left[\ket{p_1}_{i_1 }\bra{q_1}_{i_1 } \otimes I_{j_1}\right]\otimes \left[\ket{p_2}_{i_2 }\bra{q_2}_{i_2 } \otimes I_{j_2}\right] \\
& \quad \otimes \cdots \otimes \left[\ket{p_N}_{i_N}\bra{q_N}_{i_N } \otimes I_{j_N}\right] \otimes I  \cdots ,
\end{split}
\end{align}
where $a_{p_1 p_2 \cdots p_N q_1 q_2 \cdots q_N}$ are the matrix elements of the operator. Each $p_k,q_k$ index ($k \in \{ 1,2,\cdots,N\}$) is valued in $\{0,1,2\}$ and specifies an orthonormal basis vector of one black qutrit. The $\otimes I \cdots$ means that $a^{(N)}$ acts as the identity on all collections beyond the $N$th collection. Each collection is represented by square brackets. The label $N$ may be any natural number. The $(N)$ superscript reminds us of the value of $N$ for this operator. The operator $a^{(N)}$ maps $p\calh_{code}\rightarrow p\calh_{code}$. Because there exists a $K > 0$ such that $||a^{(N)} \ket{\psi}|| \leq K ||\ket{\psi}||$ for all $\ket{\psi} \in p\calh_{code}$, $a^{(N)}$ is bounded. Thus, $a^{(N)}$ maps Cauchy sequences in $p\calh_{code}$ into Cauchy sequences in $p\calh_{code}$, and Theorem \ref{thm:BLT} implies that $a^{(N)}$ is uniquely defined as a bounded operator acting on $\calh_{code}$. The $\star$-algebra $A_{code}$ is closed under hermitian conjugation and contains the identity.

A sequence of operators $\{a_n\} \in A_{code}$ converges strongly to an operator in $\calb(\calh_{code})$ if and only if $\lim_{n \rightarrow \infty} a_n \ket{\Psi}$ converges for all $\ket{\Psi} \in \calh_{code}$. The $\star$-algebra $A_{code}$ is not closed under strong limits. The von Neumann algebra $M_{code}$ is defined to be the closure of $A_{code}$ in the strong operator topology. We construct $M_{code}$ from all strongly converging limits of sequences in $A_{code}$. 
In topology, to construct the closure of a set, it is necessary, but generally \emph{not} sufficient, to include limits of converging sequences \cite{ReedSimon}. We must also include limits of nets, which are more general than sequences. However, it is possible to show that every operator in $M_{code}$ can be written as a strong limit of a sequence in $A_{code}$. In the next section, we show that the set $S \subset \calb(\calh_{code})$ of bounded operators that are strong limits of sequences in $A_{code}$ is the smallest strongly closed subset of $\calb(\calh_{code})$ that contains $A_{code}$, which implies that $M_{code} = S$. This is because
\begin{itemize}
	\item $S$ is equal to the commutant of a $\star$-algebra that contains the identity, which is a von Neumann algebra \cite{Witten:2018zxz}. Because $S$ is a von Neumann algebra, $S$ is strongly closed. 
	\item Any strongly closed subset of $\calb(\calh_{code})$ that contains $A_{code}$ must contain $S$ because $S$ is defined to only contain all strongly convergent sequences in $A_{code}$.
\end{itemize}
We provide explicit details in the next subsection.

\subsection{The commutant of $A_{code}$ and $M_{code}$}

\label{sec:mcodeprime}
In this section, we explicitly describe the commutant of $A_{code}$, which is denoted by $A_{code}^\prime$. 
Then, we demonstrate that every operator in $M_{code}$ may be written as a strongly convergent sequence of operators in $A_{code}$.

An orthonormal basis of $p\calh_{code}$ is an orthonormal basis of $\calh_{code}$. To see this, let $\ket{\Phi} \in \calh_{code}$. Let $\{\ket{\phi_n}\} \in p\calh_{code}$ be a sequence that converges to $\ket{\Phi}$. Suppose that $\ket{\Phi}$ is orthogonal to every orthonormal basis vector of $p\calh_{code}$. Using Definition \ref{def:orthonormalbasis}, we need to show that $\ket{\Phi} = 0$. Indeed, $\braket{\phi_n|\Phi} = 0 \,\, \forall n \in \mathbb{N}$, so $\braket{\Phi|\Phi} = 0$. Hence, $\ket{\Phi} = 0$.

Thus, we may define an orthonormal basis of $\calh_{code}$ where each basis vector is a finite linear combination of the vectors given in equation \eqref{eq:codeprehbasis}. We will choose an orthonormal basis $\ket{e_i}, \, i \in \mathbb{N}$ such that the first $9^\ell$ orthonormal basis vectors in the sequence $\{\ket{e_i}\}$ span the subspace of $p\calh_{code}$ where the qutrit pairs in the $(\ell+1)$th collection and beyond are in the reference state $\ket{\lambda}$.

A consequence of Theorem \ref{thm:strongoperatorconvergence} is that any operator $\calo \in \calb(\calh_{code})$ may be written as the following strong limit:
\begin{equation} \calo = \slim_{n \rightarrow \infty} \calo_n, \quad
\calo_n := \sum_{i =1}^{9^n} \sum_{j =1}^{9^n} \braket{e_{i}|\calo|e_{j}}\ket{e_{i}}\bra{e_{j}}.
\label{eq:on} 
\end{equation}
Each operator $\calo_n$ acts as the projector onto $\ket{\lambda}$ on the qutrits in the $(n+1)$th collection and beyond. Each $\calo_n$ may be written as 
\begin{align}
\label{eq:onoperator}
\begin{split}
\mathcal{O}_n = &\sum_{p_1 \cdots p_{n} q_1 \cdots q_{n} r_1 \cdots r_{n} s_1 \cdots s_{n}} \left( \mathcal{O}^n_{p_1 \cdots p_{n} q_1  \cdots q_{n} r_1  \cdots r_{n} s_1  \cdots s_{n}} 
\right.\\ \times 
&\left.\left[\ket{p_1}_{i_1 }\bra{q_1}_{i_1 } \otimes \ket{r_1}_{j_1 }\bra{s_1}_{j_1}\right]  
\otimes \cdots \otimes \left[\ket{p_{n}}_{i_{n} }\bra{q_{n}}_{i_{n} } \otimes \ket{r_{n}}_{j_{n} }\bra{s_{n}}_{j_{n}}\right]   \otimes \ket{\lambda}\bra{\lambda}  \cdots \right),
\end{split}
\end{align}
where the coefficient of each term of the sum is defined as
\begin{align}
\label{eq:matrixelements}
\begin{split} &\mathcal{O}^n_{p_1 \cdots p_{n} q_1  \cdots q_{n} r_1  \cdots r_{n} s_1  \cdots s_{n}} := \\ &\left(\left[\bra{p_1}_{i_1}\bra{r_1}_{j_1}\right]\otimes\cdots\otimes\left[\bra{p_{n}}_{i_{n}}\bra{r_{n}}_{j_{n}}\right] \otimes \bra{\lambda}\cdots\right)\calo\left(\left[\ket{q_1}_{i_1}\ket{s_1}_{j_1}\right]\otimes\cdots\otimes\left[\ket{q_{n}}_{i_{n}}\ket{s_{n}}_{j_1}\right]\otimes \ket{\lambda} \cdots \right) .
\end{split}
\end{align}
The $\otimes\ket{\lambda}\bra{\lambda}\cdots$ means that in all collections past the $n$th collection, $\calo_n$ acts as the projector $\ket{\lambda}\bra{\lambda}$. Likewise, $\otimes \ket{\lambda}\cdots$  means that in every collection past the $n$th collection, the qutrits are in the state $\ket{\lambda}$. Each of the indices $p_k$,$q_k$,$r_k$,$s_k$ ($k \in \{ 1,2,\ldots,n\}$) are valued in $\{0,1,2\}$ and denote an orthonormal basis vector of a single qutrit.

For each $\calo_n$, define the following:
\begin{align}
\begin{split}
\hat{\calo}_n := &\sum_{p_1 \cdots p_{n} q_1 \cdots q_{n} r_1 \cdots r_{n} s_1 \cdots s_{n}} \left( \mathcal{O}^n_{p_1 \cdots p_{n} q_1  \cdots q_{n} r_1  \cdots r_{n} s_1  \cdots s_{n}} 
\right.\\ \times 
&\left.\left[\ket{p_1}_{i_1 }\bra{q_1}_{i_1 } \otimes \ket{r_1}_{j_1 }\bra{s_1}_{j_1}\right]  
\otimes \cdots \otimes \left[\ket{p_{n}}_{i_{n} }\bra{q_{n}}_{i_{n} } \otimes \ket{r_{n}}_{j_{n} }\bra{s_{n}}_{j_{n}}\right]   \otimes I  \cdots \right).
\end{split}
\end{align}
The projector $\ket{\lambda}\bra{\lambda}$ in equation \eqref{eq:onoperator} has been replaced by the identity operator. For any vector $\ket{\psi} \in p\calh_{code}$, we have $\lim_{n \rightarrow \infty} \calo_n \ket{\psi} = \lim_{n \rightarrow \infty} \hat{\calo}_n \ket{\psi}$.  Also, $||\calo_n|| = ||\hat{\calo}_n|| \, \forall n \in \mathbb{N}$, so the sequence of norms $\{||\hat{\calo}_n||\}$ is bounded because the sequence of norms $\{||\calo_n||\}$ is bounded. Because $p\calh_{code}$ is dense in $\calh_{code}$, $\lim_{n \rightarrow \infty} \hat{\calo}_n$ converges strongly to $\calo$ by Theorem \ref{thm:convergence}.

Now, we assume that $\calo \in A_{code}^\prime$. The commutant $A_{code}^\prime$ is a von Neumann algebra because it is the commutant of a $\star$-algebra containing the identity. This assumption restricts what the matrix elements of equation \eqref{eq:matrixelements} can be. By considering the commutator of $\calo$ with operators in $A_{code}$, one finds that $\hat{\calo}_n$ can be written as
\begin{equation}
\label{eq:acodeprime}
\hat{\calo}_n = \sum_{r_1 \cdots r_{n} s_1 \cdots s_{n}} \left( \hat{\mathcal{O}}^n_{r_1  \cdots r_{n} s_1  \cdots s_{n}} 
\right. \times 
\left.\left[I_{i_1 } \otimes \ket{r_1}_{j_1 }\bra{s_1}_{j_1}\right]  
\otimes \cdots \otimes \left[I_{i_{n} } \otimes \ket{r_{n}}_{j_{n} }\bra{s_{n}}_{j_{n}}\right]   \otimes I  \cdots \right),
\end{equation}
for some coefficients $\hat{\mathcal{O}}^n_{r_1  \cdots r_{n} s_1  \cdots s_{n}}$. Thus, we have demonstrated that every operator $\calo \in A_{code}^\prime$ can be expressed as $\calo = \slim_{n \rightarrow \infty} \hat{\calo}_n$ where each $\hat{\calo}_n$ may be written as above. Furthermore, every such strong limit is clearly in $A_{code}^\prime$.

By comparing equation \eqref{eq:acodeprime} with equation \eqref{eq:opinacode}, it is clear the set of operators in $A_{code}$ together with strong limits of sequences in $A_{code}$ (which we called $S$ in the previous subsection) is a von Neumann algebra. In fact, it is the smallest strongly closed subset of $\calb(\calh_{code})$ containing $A_{code}$, which is $M_{code}$ by definition. This is because the strong closure of $A_{code}$ must at least contain all strongly convergent sequences of operators in $A_{code}$. Hence, every operator in $M_{code}$ may be written as a strong limit of a sequence in $A_{code}$.

Because $M_{code} = \overline{A_{code}} = A_{code}^{\prime \prime}$, we have that $M_{code}^\prime = A_{code}^{\prime \prime \prime} = \overline{A_{code}^{\prime}} = A_{code}^\prime$.
Thus, we see that $M_{code}^\prime$ may be constructed in the same way as $M_{code}$, except operators in $M_{code}^\prime$ only act nontrivially on the $j$ qutrit in Figure \ref{fig:twotofourexample}. 

From our explicit construction of $M_{code}^\prime$, we see that $M_{code}$ and $M_{code}^\prime$ are both factors as $M_{code} \cap M_{code}^\prime$ only consists of scalar multiples of the identity.


\subsection{Definition of $M_{phys}$ and $M_{phys}^\prime$}

Recall that under the isometry in equation \eqref{eq:onefactorisometry}, the code reference state $\ket{\lambda}$ on the black qutrits $i,j$ in Figure \ref{fig:twotofourexample} is mapped to the state of four white qutrits where the qutrit pairs $\tilde{a},\tilde{d}$ and $\tilde{b},\tilde{e}$ are each in the state $\ket{\lambda}$. Thus, both the physical and code pre-Hilbert spaces consist of states of infinitely many qutrit pairs, all but finitely many of which are in the reference state $\ket{\lambda}$. It follows that $\calh_{code}$ and $\calh_{phys}$ are constructed in the exact same way. We can define a von Neumann algebra $M_{phys} \subset \calb(\calh_{phys})$ acting on the white qutrits $\tilde{a},\tilde{b}$ in each collection in the same way we defined $M_{code}$ to act on the black qutrit $i$. Likewise, the commutant of $M_{phys}$, denoted by $M_{phys}^\prime$, acts on white qutrits $\tilde{d},\tilde{e}$. Our setup is summarized in Figure \ref{fig:qutritcodes}.

\section{Definition of the tensor network map}
\label{sec:tensornetworkmap}
Having defined $M_{code}$ and $M_{phys}$, we define a linear map from $M_{code}$ into $M_{phys}$. An operator $\calo \in M_{code}$ is mapped to $\tilde{\calo} \in M_{phys}$. We want the following to hold for all $\ket{\Psi} \in \calh_{code}$:
\begin{equation} 
u \calo \ket{\Psi} = \tilde{\calo} u \ket{\Psi}, \quad u \calo^\dagger \ket{\Psi} = \tilde{\calo}^\dagger u \ket{\Psi} .
\label{eq:niceproperty}
\end{equation}
We now describe how to construct this map (which we call the ``tensor network map,'' not to be confused with the map $u$).

\subsection{How the tensor network map acts on $A_{code}$}

We first define how the tensor network map acts on operators in $A_{code}$ before generalizing its definition to $M_{code}$. The operator $a^{(N)}$ in equation \eqref{eq:opinacode} is mapped to $\tilde{a}^{(N)}$, an operator that acts on $\calh_{phys}$. The result is
\begin{align}
\begin{split}
\label{eq:opinaphys}
\tilde{a}^{(N)} = \sum_{p_1 \cdots p_N q_1  \cdots q_N} & a_{p_1  \cdots p_N q_1  \cdots q_N} \times \left[U_{\tilde{a}_1 \tilde{b}_1}\ket{p_1}_{\tilde{a}_1}\bra{q_1}_{\tilde{a}_1}  U_{\tilde{a}_1\tilde{b}_1}^\dagger \otimes I_{\tilde{d}_1\tilde{e}_1}\right] \\
& \otimes \cdots \otimes \left[U_{\tilde{a}_N \tilde{b}_N}\ket{p_N}_{\tilde{a}_N}\bra{q_N}_{\tilde{a}_N}  U_{\tilde{a}_N \tilde{b}_N}^\dagger \otimes I_{\tilde{d}_N \tilde{e}_N}\right] \otimes I \cdots ,
\end{split}
\end{align}
where $U$ is defined in equation \eqref{eq:unitary}, and the subscripts refer to the specific white qutrits that $U$ is acting on (see Figure \ref{fig:twotofourexample}). Given equation \eqref{eq:opinaphys}, which shows how $\tilde{a}^{(N)}$ acts on vectors in $p\calh_{phys}$, the domain of $\tilde{a}^{(N)}$ may be extended to all of $\calh_{phys}$ by demanding that $\tilde{a}^{(N)}$ is a bounded operator and invoking Theorem \ref{thm:BLT}. Because $\tilde{a}^{(N)}$ acts trivially on the qutrits $\tilde{d},\tilde{e}$ in each collection, $\tilde{a}^{(N)} \in M_{phys}$. 

Equation \eqref{eq:opinaphys} simply amounts to applying the map in equation \eqref{eq:onefactormap} for a finite number of collections. It follows that for $a,b \in A_{code}$, $\alpha,\beta \in \mathbb{C}$, and $\ket{\Psi} \in \calh_{code}$, the tensor network map has the following properties:
\begin{align}
& 1.\ \text{Bulk Reconstruction}:\quad 
u a \ket{\Psi} = \tilde{a} u \ket{\Psi}, \label{eq:acodebulkrecon}
\\
& 2.\ \text{Commutativity with hermitian conjugation}:\quad 
\widetilde{a^\dagger} = \tilde{a}^\dagger, \label{eq:acodecommutedagger}
\\
& 3.\ \text{Commutativity with multiplication}:\quad 
\widetilde{ab} = \tilde{a}\tilde{b}, \label{eq:acodecommutemultiplication}
\\
& 4.\ \text{Linearity}:\quad 
\widetilde{\alpha a + \beta b} = \alpha \tilde{a} + \beta \tilde{b}, \label{eq:acodelinearity}
\\
& 5.\ \text{Norm preservation}:\quad 
||a|| = ||\tilde{a}|| \label{eq:acodenormpreserve}.
\end{align}

We will prove these properties for all operators in $M_{code}$ in Section \ref{sec:tensornetworkproperties}.

\subsection{How the tensor network map acts on $M_{code}$}

Now that we specified how the tensor network map acts on $A_{code}$, we need to specify how it acts on $M_{code}$. Let $\{a_n\} \in A_{code}$ be a strongly convergent sequence of operators. The image of each $a_n$ under the tensor network map is $\tilde{a}_n \in \calb(\calh_{phys})$. We will show that $\{\tilde{a}_n\}$ is a strongly convergent sequence. Then, we will extend the definition of the tensor network map by saying that the strong limit $(\slim_{n \rightarrow \infty} a_n) \in M_{code}$ is mapped to the strong limit $(\slim_{n \rightarrow \infty} \tilde{a}_n) \in M_{phys}$. We will then prove that this map satisfies equation \eqref{eq:niceproperty}.

The fact that $\slim_{n \rightarrow \infty} a_n$ converges means that the sequence of norms $\{||\tilde{a}_n||\}$ is bounded from above because $||a_n|| = ||\tilde{a}_n|| \, \forall n \in \mathbb{N}$. From Theorem \ref{thm:convergence}, if $\lim_{n \rightarrow \infty} \tilde{a}_n \ket{\psi}$ converges for all $\ket{\psi} \in p\calh_{phys}$, then $\lim_{n \rightarrow \infty} \tilde{a}_n \ket{\Psi}$ converges for all $\ket{\Psi} \in \calh_{phys}$ since $p\calh_{phys}$ is dense in $\calh_{phys}$.  The next theorem is necessary to show that $\lim_{n \rightarrow \infty} \tilde{a}_n \ket{\psi}$ converges for all $\ket{\psi} \in p\calh_{phys}$.

\begin{thm}
	\label{thm:codephysequiv}
	For any two vectors $\ket{\tilde{\psi}_1}, \ket{\tilde{\psi}_2} \in p\calh_{phys}$, we may define a finite number of vectors $\ket{\eta_i},\ket{\chi_i} \in p \calh_{code}$, ($i \in \{ 1,2,\ldots,Q \}$ for some $Q \in \mathbb{N}$) such that for any operator $\tilde{a}^{(N)} \in M_{phys}$ that may be written as the tensor network map image of some $a^{(N)} \in A_{code}$, we have that
	\begin{equation} \label{eq:phystocode}
	\braket{\tilde{\psi}_1|\tilde{a}^{(N)}|\tilde{\psi}_2} = \sum_{i = 1}^Q\braket{\eta_i|a^{(N)}|\chi_i}.
	\end{equation}
	Furthermore, if $\ket{\tilde{\psi}_1} = \ket{\tilde{\psi}_2}$, then we may take $\ket{\eta_i} = \ket{\chi_i} \, \, \forall i$. 
\end{thm}
\begin{proof}
Choose $M \in \mathbb{N}$ such that for both $\ket{\tilde{\psi}_1}$ and $\ket{\tilde{\psi}_2}$, the qutrits in the $(M+1)$th collection and beyond are in the reference state $\ket{\lambda}$. Consider the following set of orthonormal vectors: 
\begin{align}
\begin{split} 
\ket{\{r,\ell,s\}} =& \left[U_{\tilde{a}_1 \tilde{b}_1}U_{\tilde{d}_1 \tilde{e}_1}\ket{r_1}_{\tilde{a}_1}\ket{\ell_1}_{\tilde{d}_1}  \ket{s_1}_{\tilde{b}_1  \tilde{e}_1}\right] \otimes \\
& \quad \cdots \otimes \left[ U_{\tilde{a}_M \tilde{b}_M}U_{\tilde{d}_M \tilde{e}_M} \ket{r_M}_{\tilde{a}_M} \ket{\ell_M}_{\tilde{d}_M}  \ket{s_M}_{\tilde{b}_M  \tilde{e}_M}\right] \otimes \ket{\lambda \lambda} \cdots , 
\end{split} 
\end{align}
where the labels $\tilde{a}_k$,$\tilde{b}_k$,$\tilde{d}_k$,$\tilde{e}_k$ ($k \in \{1,2,\ldots,M\}$) refer to the qutrits in the $k$th collection (see Figure \ref{fig:twotofourexample}). Each $r_k$ and $\ell_k$ index is valued in $\{0,1,2\}$ and specifies an orthonormal basis vector of one qutrit. Each $s_k$ index is valued in $\{0,1,2,\ldots,8\}$ and specifies an orthonormal basis vector of two qutrits. The $\ket{\lambda \lambda}\cdots$ means that in all collections past the $M$th collection, the qutrits are in the physical reference state $\ket{\lambda}_{\tilde{a} \tilde{d}} \ket{\lambda}_{\tilde{b} \tilde{e}}$.

We may then write $\ket{\tilde{\psi}_1}, \ket{\tilde{\psi}_2}$ as finite linear combinations of the above vectors:
\begin{equation}
\ket{\tilde{\psi}_1} = \sum_{\{s\}} \sum_{ \{r,\ell\} } c^1_{\{r,\ell,s\}} \ket{\{r,\ell,s\}},
\quad
\ket{\tilde{\psi}_2} = \sum_{\{s\}} \sum_{ \{r,\ell\} } c^2_{\{r,\ell,s\}} \ket{\{r,\ell,s\}},
\end{equation}
where $c^1_{\{r,\ell,s\}}$ and $c^2_{\{r,\ell,s\}}$ are $\mathbb{C}$-valued coefficients.
Note that  $\braket{\{r,\ell,s\}|\tilde{a}^{(N)}|\{r^\prime,\ell^\prime,s^\prime\}} = 0$ if $s_k \neq s^\prime_k$ for any $k \in \{ 1,2,\ldots,M\}$. Thus, we write
\begin{equation}
\braket{\tilde{\psi}_1|\tilde{a}^{(N)}|\tilde{\psi}_2} = 
\sum_{\{s\}} \sum_{ \{r^\prime,\ell^\prime\} }
\sum_{ \{r,\ell\} } (c^1_{\{r^\prime,\ell^\prime,s\}})^* \braket{\{r^\prime,\ell^\prime,s\}|\tilde{a}^{(N)}|\{r,\ell,s\}}
 c^2_{\{r,\ell,s\}}.
\end{equation}
To calculate $\braket{\{r^\prime,\ell^\prime,s\}|\tilde{a}^{(N)}|\{r,\ell,s\}}$, we must calculate how each term in the sum in equation \eqref{eq:opinaphys} acts on each collection separately. The next three equations apply for a single collection. For simplicity, we have suppressed the subscripts labeling the collection.

\begin{align}
\begin{split}
1.\ &\left[U_{\tilde{a} \tilde{b}}U_{\tilde{d} \tilde{e}}\ket{r^\prime}_{\tilde{a}}  \ket{\ell^\prime}_{\tilde{d}}  \ket{s^\prime}_{\tilde{b}  \tilde{e}}\right]^\dagger 
\left[U_{\tilde{a} \tilde{b}}\ket{p}_{\tilde{a}}\bra{q}_{\tilde{a}}  U_{\tilde{a} \tilde{b}}^\dagger \otimes I_{\tilde{d} \tilde{e}}\right]
\left[U_{\tilde{a} \tilde{b}}U_{\tilde{d} \tilde{e}}\ket{r}_{\tilde{a}}\ket{\ell}_{\tilde{d}}  \ket{s}_{\tilde{b}  \tilde{e}}\right]
\\
&\quad\quad=\left[\bra{r^\prime}_{\tilde{a}}\bra{\ell^\prime}_{\tilde{d}}  \bra{s^\prime}_{\tilde{b}  \tilde{e}}\right]
\left[\ket{p}_{\tilde{a}}\bra{q}_{\tilde{a}}  \otimes I_{\tilde{b} \tilde{d} \tilde{e}}\right]
\left[\ket{r}_{\tilde{a}}\ket{\ell}_{\tilde{d}}  \ket{s}_{\tilde{b}  \tilde{e}}\right],
\end{split}\\
\begin{split}
2.\ &\left[U_{\tilde{a} \tilde{b}}U_{\tilde{d} \tilde{e}}\ket{r^\prime}_{\tilde{a}}\ket{\ell^\prime}_{\tilde{d}}  \ket{s^\prime}_{\tilde{b}  \tilde{e}}\right]^\dagger
I_{\tilde{a} \tilde{b} \tilde{d} \tilde{e}} \left[U_{\tilde{a} \tilde{b}}U_{\tilde{d} \tilde{e}}\ket{r}_{\tilde{a}}\ket{\ell}_{\tilde{d}}  \ket{s}_{\tilde{b}  \tilde{e}}\right]
=\left[\bra{r^\prime}_{\tilde{a}}\bra{\ell^\prime}_{\tilde{d}}  \bra{s^\prime}_{\tilde{b}  \tilde{e}}\right]
\left[\ket{r}_{\tilde{a}}\ket{\ell}_{\tilde{d}}  \ket{s}_{\tilde{b}  \tilde{e}}\right] ,
\end{split} \\
\begin{split}
3.\ &\bra{\lambda}_{\tilde{a} \tilde{d}}\bra{\lambda}_{\tilde{b} \tilde{e}}
\left[U_{\tilde{a} \tilde{b}}\ket{p}_{\tilde{a}}\bra{q}_{\tilde{a}}  U_{\tilde{a} \tilde{b}}^\dagger \otimes I_{\tilde{d} \tilde{e}}\right]
\ket{\lambda}_{\tilde{a} \tilde{d}}\ket{\lambda}_{\tilde{b} \tilde{e}}=\bra{\lambda}_{\tilde{a} \tilde{d}}
\left[\ket{p}_{\tilde{a}}\bra{q}_{\tilde{a}}  \otimes I_{\tilde{b} \tilde{d} \tilde{e}}\right]
\ket{\lambda}_{\tilde{a} \tilde{d}}.
\end{split}
\end{align}

Next, we define the following vectors in $p\calh_{code}$
\begin{equation} 
\ket{\{r,\ell\}}  := \left[\ket{r_1}_{i_1} \ket{\ell_1}_{j_1}\right] \otimes \left[\ket{r_2}_{i_2} \ket{\ell_2}_{j_2}\right] \otimes \cdots \otimes \left[\ket{r_M}_{i_M} \ket{\ell_M}_{j_M}\right] \otimes \ket{\lambda} \cdots. 
\end{equation}
It follows that
\begin{equation}
\braket{\tilde{\psi}_1|\tilde{a}^{(N)}|\tilde{\psi}_2} = 
\sum_{\{s\}} \sum_{ \{r^\prime,\ell^\prime\} }
\sum_{ \{r,\ell\} } (c^1_{\{r^\prime,\ell^\prime,s\}})^* \braket{\{r^\prime,\ell^\prime\}|a^{(N)}|\{r,\ell\}}
c^2_{\{r,\ell,s\}}.
\end{equation}
Then, we can define the new vectors in $p\calh_{code}$
\begin{equation}
\ket{\chi_{ \{s\} }} := \sum_{ \{r\} \{\ell \}}\ket{\{r,\ell\}} c^2_{\{r,\ell,s\}}, \quad
\ket{\eta_{ \{s\} }} := \sum_{ \{r\} \{\ell \}}\ket{\{r,\ell\}} c^1_{\{r,\ell,s\}},
\end{equation}
so that $\braket{\tilde{\psi}_1|\tilde{a}^{(N)}|\tilde{\psi}_2}$ can be expressed as
\begin{equation}
\braket{\tilde{\psi}_1|\tilde{a}^{(N)}|\tilde{\psi}_2} = 
\sum_{\{s\}} \braket{\eta_{\{s\}}|a^{(N)}|\chi_{\{s\}}}.
\end{equation}
This demonstrates that we can express $\braket{\tilde{\psi}_1|\tilde{a}^{(N)}|\tilde{\psi}_2}$ as in equation \eqref{eq:phystocode} for $Q = 9^M$.
\end{proof}
Given any $\ket{\tilde{\psi}} \in p\calh_{phys}$, Theorem \ref{thm:codephysequiv} asserts that we may choose a finite family of vectors $\ket{\psi_i} \in p\calh_{code}$ ($i \in \{ 1,2,\ldots,Q\}$ for some $Q \in \mathbb{N}$ ) such that, for any $n,m \in \mathbb{N}$, \begin{equation}||(\tilde{a}_n - \tilde{a}_m)\ket{\tilde{\psi}}||^2 = \sum_{i = 1}^Q||(a_n - a_m)\ket{\psi_i}||^2.\end{equation} This means that if $\{a_n \ket{\psi_i}\}$ is a Cauchy sequence for each $i$, (which it is by assumption) then $\{\tilde{a}_n \ket{\tilde{\psi}}\}$ is also a Cauchy sequence. This shows that $\lim_{n \rightarrow \infty} \tilde{a}_n\ket{\tilde{\psi}}$ converges for any $\ket{\tilde{\psi}} \in p\calh_{phys}$.   
 Thus, the strong limit $\slim_{n \rightarrow \infty}\tilde{a}_n$ exists and
defines an operator, which is the \emph{definition} of the image under the tensor network map of the strong limit $\slim_{n \rightarrow \infty} a_n$. By the definition of $M_{phys}$, it follows that $\slim_{n \rightarrow \infty} \tilde{a}_n \in M_{phys}$.

 Suppose that the sequences $\{a_n\} \in A_{code}$ and $\{b_n\} \in A_{code}$ converge strongly to the same operator $\calo$. Suppose that $\slim_{n \rightarrow \infty} \tilde{a}_n = \tilde{\calo}_1$ and $\slim_{n \rightarrow \infty} \tilde{b}_n=\tilde{\calo}_2$. Then $\slim_{n \rightarrow \infty} (a_n - b_n) = 0$, which implies that $\slim_{n \rightarrow \infty} \widetilde{a_n-b_n} = \slim_{n \rightarrow \infty} (\tilde{a}_n - \tilde{b}_n) = 0$. Hence, $\tilde{\calo}_1 - \tilde{\calo}_2 = 0$. Thus, the tensor network map is a well-defined map from $M_{code}$ into $M_{phys}$.

\subsection{How the tensor network map acts on $M_{code}^\prime$}

By construction, the tensor network map is a map from operators in $M_{code}$ into $M_{phys}$. Due to the symmetry of the tensor network in Figure \ref{fig:twotofourexample}, we can also define the tensor network map on operators in $M_{code}^\prime$, which are mapped into $M_{phys}^\prime$ in a completely analogous way. 

\section{Properties of the Tensor Network Map}
\label{sec:tensornetworkproperties}
In this section, we prove that equations $\eqref{eq:acodebulkrecon}$ to $\eqref{eq:acodenormpreserve}$ hold for all operators in $M_{code}$.

\subsection{Theorems on strong and weak convergence}

The following theorems will be useful in proving some properties of the tensor network map.

\begin{thm}
	\label{thm:strongweakconvergence}
	Suppose that for a sequence $\{a_n\} \in A_{code}$, $\lim_{n \rightarrow \infty}\braket{\Psi|a_n|\Phi}=0$ for any $\ket{\Psi},\ket{\Phi} \in \calh_{code}$. Suppose that the sequence of norms $\{||a_n||\}$ is bounded from above. Let $\tilde{a}_n$ be the image under the tensor network map of $a_n$.  Then $\lim_{n \rightarrow \infty}\braket{\tilde{\Theta}|\tilde{a}_n|\tilde{\Phi}} = 0$ for any $\ket{\tilde{\Theta}},\ket{\tilde{\Phi}} \in \calh_{phys}$.
\end{thm} 
\begin{proof}


	Let $\{\ket{\tilde{\theta}_\ell}\},\{\ket{\tilde{\phi}_m}\} \in p\calh_{phys}$ be Cauchy sequences that converge to $\ket{\tilde{\Theta}},\ket{\tilde{\Phi}}\in \calh_{phys}$ respectively. We may compute
	\begin{equation} 
	|\braket{\tilde{\Theta}|\tilde{a}_n|\tilde{\Phi}}| \leq |\braket{\tilde{\Theta}|\tilde{a}_n|\tilde{\Phi} - \tilde{\phi}_m}|+|\braket{\tilde{\Theta}-\tilde{\theta}_\ell|\tilde{a}_n|\tilde{\phi}_m}|+|\braket{\tilde{\theta}_\ell|\tilde{a}_n|\tilde{\phi}_m}|,
	\end{equation}
	\begin{equation} 
	|\braket{\tilde{\Theta}|\tilde{a}_n|\tilde{\Phi}}| \leq ||\ket{\tilde{\Theta}}|| \cdot ||\tilde{a}_n|| \cdot ||\ket{\tilde{\Phi}} - \ket{\tilde{\phi}_m}||+||\ket{\tilde{\Theta}}-\ket{\tilde{\theta}_\ell}|| \cdot ||\tilde{a}_n|| \cdot ||\ket{\tilde{\phi}_m}||+|\braket{\tilde{\theta}_\ell|\tilde{a}_n|\tilde{\phi}_m}|,
	\end{equation}
	\begin{equation}
	|\braket{\tilde{\Theta}|\tilde{a}_n|\tilde{\Phi}}| \leq K_1 ||\ket{\tilde{\Phi}} - \ket{\tilde{\phi}_m}||+K_2 ||\ket{\tilde{\Theta}}-\ket{\tilde{\theta}_\ell}|| +|\braket{\tilde{\theta}_\ell|\tilde{a}_n|\tilde{\phi}_m}|,
	\label{eq:7pt3}
	\end{equation}
	where $K_1,K_2$ are some positive real numbers and we used the fact that the sequence $\{||\tilde{a}_n||\}$ is bounded from above. First, fix $m,\ell$ large enough so that the first two norms on the r.h.s. of equation \eqref{eq:7pt3} are each less than $\frac{\epsilon}{3}$. Due to Theorem \ref{thm:codephysequiv} and the assumption that $\wlim_{n \rightarrow \infty} a_n = 0$, we have that $\lim_{n \rightarrow \infty}\braket{\tilde{\theta}_\ell|\tilde{a}_n|\tilde{\phi}_m} = 0$. Hence, we can choose $N \in \mathbb{N}$ such that for $n > N$, the third norm on the r.h.s. of equation \eqref{eq:7pt3} is less than $\frac{\epsilon}{3}$. We conclude that $\lim_{n \rightarrow \infty} \braket{\tilde{\Theta}|\tilde{a}_n|\tilde{\Phi}} = 0$.
\end{proof}

\begin{thm}
\label{thm:conjugateweakconvergence}
	Let $\{a_n\} \in A_{code}$ be a strongly convergent sequence of operators. Suppose that $\slim_{n \rightarrow \infty} a_n = \calo$ for some $\calo \in M_{code}$. Then $\wlim_{n \rightarrow \infty} a_n^\dagger= \calo^\dagger$.
\end{thm}
\begin{proof}  
	Let $\ket{\Psi}, \ket{\Phi} \in \calh_{code}$. Then
	\begin{equation}  \braket{\Psi|\calo^\dagger|\Phi} =\braket{\Phi|\calo|\Psi}^* = \lim_{n \rightarrow \infty} \braket{\Phi|a_n|\Psi}^* = \lim_{n \rightarrow \infty} \braket{\Psi|a_n^\dagger|\Phi}, \end{equation}
	so the sequence of operators $\{a_n^\dagger\}$ converges weakly to $\calo^\dagger$. Recalling Theorem \ref{thm:doublecommutant}, we see explicitly that $M_{code}$ is closed under hermitian conjugation.
\end{proof}

\subsection{The tensor network map is linear}
We now demonstrate the linearity of the tensor network map. Consider two sequences of operators in $A_{code}$, $\{a_n\}$ and $\{b_n\}$, converging strongly to $\calo_1$ and $\calo_2$ respectively. Then for $\alpha, \beta \in \mathbb{C}$, $\slim_{n \rightarrow \infty}(\alpha a_n + \beta b_n) = \alpha \calo_1 + \beta \calo_2$. The image of each $a_n$ is $\tilde{a}_n$ and the image of each $b_n$ is $\tilde{b}_n$. The image of $\alpha \calo_1 + \beta \calo_2$ under the tensor network map is thus given by $\alpha \tilde{\calo}_1 + \beta \tilde{\calo}_2$. 
Hence, the tensor network map is linear when acting on all operators in $M_{code}$.

\subsection{The tensor network map commutes with hermitian conjugation}

 If $\{a_n\} \in A_{code}$ strongly converges to $\calo$, then $\wlim_{n \rightarrow \infty} a_n^\dagger=\calo^\dagger$ by Theorem \ref{thm:conjugateweakconvergence}. Each $a_n$ is mapped to $\tilde{a}_n$ under the tensor network map, and $\{\tilde{a}_n\}$ strongly converges to $\tilde{\calo} \in \calb(\calh_{phys})$. Each $a_n^\dagger$ is mapped to $\widetilde{a_n^\dagger} = \tilde{a}_n^\dagger$, and $\wlim_{n \rightarrow \infty} \tilde{a}_n^\dagger=(\tilde{\calo})^\dagger$. Since $M_{code}$ is defined from $A_{code}$ by taking strong limits, there must exist a sequence $\{b_n\} \in A_{code}$ that converges strongly to $\calo^\dagger$. Then, $\slim_{n \rightarrow \infty} \tilde{b}_n=\widetilde{\calo^\dagger}$. Note that, for any two $\ket{\Psi},\ket{\Phi} \in \calh_{code}$, $\lim_{n \rightarrow \infty}\braket{\Psi|(a_n^\dagger - b_n) |\Phi} = 0$. The sequence of norms $\{||a_n^\dagger - b_n||\}$ is bounded above because $||a_n^\dagger||=||a_n|| \, \, \forall n \in \mathbb{N}$ and $\{a_n\}$ and $\{b_n\}$ converge strongly. Furthermore, for any $\ket{\tilde{\Psi}},\ket{\tilde{\Phi}} \in \calh_{phys}$, $\lim_{n \rightarrow \infty}\braket{\tilde{\Psi}|\widetilde{(a_n^\dagger - b_n)} |\tilde{\Phi}}=\braket{\tilde{\Psi}|(\tilde{\calo}^\dagger - \widetilde{\calo^\dagger}) |\tilde{\Phi}}$. Applying Theorem \ref{thm:strongweakconvergence}, $\braket{\tilde{\Psi}|(\tilde{\calo}^\dagger - \widetilde{\calo^\dagger}) |\tilde{\Phi}} =0$, hence $\tilde{\calo}^\dagger = \widetilde{\calo^\dagger}$.

\subsection{The tensor network map commutes with multiplication}




Given $\calo_1,\calo_2 \in M_{code}$, we now show that $\widetilde{\calo_1 \calo_2} = \tilde{\calo}_1 \tilde{\calo}_2$. Let $\{a_n\} \in A_{code}$ converge strongly to $\calo_1^\dagger$. Let $\{b_n\} \in A_{code}$ converge strongly to $\calo_2$. Let $\{c_n\} \in A_{code}$ converge strongly to $\calo_1 \calo_2$. For any $\ket{\Psi},\ket{\Phi} \in \calh_{code}$,
\begin{equation} \lim_{n \rightarrow \infty} \braket{(a_n - \calo_1^\dagger)\Psi|(b_n - \calo_2) \Phi} =0, \end{equation}
which implies that
\begin{equation} \lim_{n \rightarrow \infty} \braket{\Psi|a_n^\dagger b_n - c_n |\Phi} =0.  \end{equation}
The sequence of norms $\{||a_n^\dagger b_n - c_n||\}$ is bounded. By Theorem \ref{thm:strongweakconvergence}, we have that, for all $\ket{\tilde{\Psi}},\ket{\tilde{\Phi}} \in \calh_{phys}$,
\begin{equation} \lim_{n \rightarrow \infty} \braket{\tilde{\Psi}|\tilde{a}_n^\dagger \tilde{b}_n - \tilde{c}_n |\tilde{\Phi}} =0. \end{equation}
It follows that
\begin{equation}  \braket{\tilde{\Psi}| \tilde{\calo}_1 \tilde{\calo}_2  |\tilde{\Phi}}
 =\braket{\tilde{\Psi}|  \widetilde{\calo_1 \calo_2} |\tilde{\Phi}}.  \end{equation}

\subsection{The tensor network map preserves the norm}

Consider any $\ket{\tilde{\psi}} \in p\calh_{phys}$. By Theorem \ref{thm:codephysequiv}, there exists a finite family of vectors $\ket{\psi_i} \in p \calh_{code}$, ($i \in \{ 1,2,\ldots,Q\}$ for some $Q \in \mathbb{N}$) such that for any $a \in A_{code}$,
\begin{equation} \braket{\tilde{\psi}|\tilde{a}|\tilde{\psi}} = \sum_{i = 1}^Q \braket{\psi_i|a|\psi_i}. \end{equation}
Consider a sequence $\{a_n\} \in A_{code}$ that strongly converges to $\calo \in M_{code}$. Then we have

\begin{equation} \lim_{n \rightarrow \infty}\braket{\tilde{\psi}|\tilde{a}_n|\tilde{\psi}} = \sum_{i = 1}^Q \lim_{n \rightarrow \infty} \braket{\psi_i|a_n|\psi_i}, \end{equation}
\begin{equation} \braket{\tilde{\psi}|\tilde{\calo}|\tilde{\psi}} = \sum_{i = 1}^Q  \braket{\psi_i|\calo|\psi_i}. \end{equation}
In particular, for any $\calo \in M_{code}$, we have that
\begin{equation} \braket{\tilde{\calo}\tilde{\psi}|\tilde{\calo}\tilde{\psi}} = \sum_{i = 1}^Q  \braket{\calo\psi_i|\calo\psi_i}. \end{equation}
The norms of $\ket{\tilde{\psi}}$ and $\tilde{\calo}\ket{\tilde{\psi}}$ may be expressed as
\begin{equation} ||\ket{\tilde{\psi}}|| = \sqrt{\sum_{i = 1}^Q  ||\ket{\psi}_i||^2}, \ \quad ||\tilde{\calo}\ket{\tilde{\psi}}|| = \sqrt{\sum_{i = 1}^Q  ||\calo\ket{\psi_i}||^2}. \end{equation}
Thus,
\begin{equation} ||\tilde{\calo}\ket{\tilde{\psi}}|| \leq ||\calo||\sqrt{\sum_{i = 1}^Q   ||\ket{\psi_i}||^2} = ||\calo|| \cdot ||\ket{\tilde{\psi}}||. \end{equation}
Note that we may choose $\ket{\tilde{\psi}}$ such that $Q = 1$ and $|| \calo \ket{\psi_1}|| - ||\calo|| \cdot ||\ket{\psi_1}||$ is arbitrarily close to zero. Hence, we may choose $\ket{\tilde{\psi}}$ such that $||\tilde{\calo} \ket{\tilde{\psi}}|| - ||\calo||\cdot||\ket{\tilde{\psi}}||$ is arbitrarily close to zero. It follows from the fact that $p\calh_{phys}$ is dense in $\calh_{phys}$ and Theorem \ref{thm:BLT} that
\begin{equation} ||\tilde{\calo}|| = ||\calo||. \end{equation} 

\subsection{The tensor network map satisfies bulk reconstruction}
\begin{thm}
	Let $\calo \in M_{code}$. Let $\tilde{\calo} \in M_{phys}$ be the image of $\calo$ under the tensor network map. Let $\ket{\Psi} \in \calh_{code}$. Then
	\begin{equation} u \calo \ket{\Psi} = \tilde{\calo} u \ket{\Psi}. \label{eq:mapproperty} \end{equation}
\end{thm}
\begin{proof}
Let $\{a_n\} \in A_{code}$ be a sequence that converges strongly to $\calo \in M_{code}$. Let $\tilde{a}_n \in M_{phys}$ be the image under the tensor network map of $a_n$ for every $n \in \mathbb{N}$. By the definition of the tensor network map, $\slim_{n \rightarrow \infty} \tilde{a}_n = \tilde{\calo}$. It follows that
	\begin{equation}  \tilde{\calo}u \ket{\Psi} = \lim_{n \rightarrow \infty} \tilde{a}_n u \ket{\Psi} = \lim_{n \rightarrow \infty} u a_n \ket{\Psi} = u \lim_{n \rightarrow \infty} a_n \ket{\Psi} = u \calo \ket{\Psi}. \end{equation}
\end{proof}

This theorem demonstrates the bulk reconstruction property of the tensor network map. 
We can linearly map a given operator $\calo \in M_{code}$ to an operator $\tilde{\calo} \in M_{phys}$ such that for all $\ket{\Psi} \in \calh_{code}$, 
\begin{equation} 
u \calo \ket{\Psi} = \tilde{\calo} u \ket{\Psi}, \quad u \calo^\dagger \ket{\Psi} = \widetilde{\calo^\dagger} u \ket{\Psi} = \tilde{\calo}^\dagger u \ket{\Psi} .
\end{equation} By the symmetry of the tensor network in Figure \ref{fig:twotofourexample}, any operator $\calo^\prime \in M_{code}^\prime$ can be linearly mapped to $\tilde{\calo}^\prime \in M_{phys}^\prime$ such that for all $\ket{\Psi} \in \calh_{code}$, 
\begin{equation} 
u \calo^\prime \ket{\Psi} = \tilde{\calo}^\prime u \ket{\Psi}, \quad u \calo^{\prime \, \dagger} \ket{\Psi} = \widetilde{\calo^{\prime \, \dagger}} u \ket{\Psi}
= \tilde{\calo}^{\prime \, \dagger} u \ket{\Psi} .
\end{equation}

\section{Cyclic and separating vectors}

In this section we identify a set of cyclic and separating vectors with respect to $M_{code}$ that is dense in $\calh_{code}$. Then, we prove that all cyclic and separating vectors with respect to $M_{code}$ are mapped to cyclic and separating vectors with respect to $M_{phys}$ via the isometry $u$. This shows that our infinite-dimensional QECC satisfies the assumptions of Theorem \ref{thm:maintheorem}.

\label{sec:cyclicseparating}
\begin{thm}
	Cyclic and separating vectors with respect to $M_{code}$ are dense in $\calh_{code}$. 
\end{thm}
\begin{proof}
	Since $p\calh_{code}$ is dense in $\calh_{code}$, any vector in $\calh_{code}$ is arbitrarily close to a vector in $p\calh_{code}$, which may be denoted as $\ket{\psi} \otimes \ket{\lambda}\cdots$, where $\ket{\psi}$ is a vector in a finite-dimensional Hilbert space $\calh$ that consists of finitely many pairs of qutrits. We may write $\calh = \calh_{i} \otimes \calh_{j}$ where $\calh_{i}$ consists of the black qutrits labeled by $i$ (see Figures \ref{fig:twotofourexample} and \ref{fig:qutritcodes}), and $\calh_{j}$ consists of the black qutrits labeled by $j$. The vector $\ket{\psi}$ is arbitrarily close to a vector of maximal Schmidt number (with respect to this factorization), which we will denote by $\ket{\psi^\prime}$. Hence, any vector in $\calh_{code}$ is arbitrarily close to a vector of the form $\ket{\psi^\prime} \otimes \ket{\lambda}\cdots \in \calh_{code}$ where $\ket{\psi^\prime}$ has maximal Schmidt number under the factorization  $\calh = \calh_{i} \otimes \calh_{j}$, so we just need to show that such vectors are cyclic and separating.
	
	The vector $\ket{\psi^\prime} \otimes \ket{\lambda}\cdots \in \calh_{code}$ is cyclic with respect to $M_{code}$ because operators in $A_{code} \subset M_{code}$ may act on it to obtain any vector in $p\calh_{code}$, and $p\calh_{code}$ is dense in $\calh_{code}$. Furthermore, $\ket{\psi^\prime} \otimes \ket{\lambda}\cdots$ is certainly separating with respect to $A_{code}$ as one can see from the definition of $A_{code}$ in equation \eqref{eq:opinacode}. To see that  $\ket{\psi^\prime} \otimes \ket{\lambda}\cdots$ is separating with respect to all of $M_{code}$,  note that the same logic as above implies that $\ket{\psi^\prime} \otimes \ket{\lambda}\cdots$ is cyclic with respect to $M_{code}^\prime$. Hence, $\ket{\psi^\prime} \otimes \ket{\lambda}\cdots$ is separating with respect to $M_{code}$.
	\end{proof}
	
	\begin{proof}[Alternative Proof]		
	We now give an alternative and more explicit proof of the fact that $\ket{\psi^\prime} \otimes \ket{\lambda}\cdots$ is separating with respect to all of $M_{code}$. Given a sequence $\{a_n\} \in A_{code}$ that strongly converges to $\calo \in M_{code}$ we need to show that $\calo (\ket{\psi^\prime} \otimes \ket{\lambda}\cdots) = 0$ implies that $\calo$ annihilates every vector in $p\calh_{code}$ (which would imply that $\calo$ annihilates every Cauchy sequence and hence every vector in $\calh_{code}$).
	
	First, we will construct a suitable (yet overcomplete) basis of $p\calh_{code}$. Let us assume that $\ket{\psi^\prime}$ is a state of the black qutrits in the first $M$ collections. Since $\ket{\psi^\prime}$ is a vector in a finite-dimensional factorized Hilbert space with maximal Schmidt number, we may write it as
	\begin{equation}  \ket{\psi^\prime} = \sum_{k = 1}^{3^M} \alpha_k \ket{e_k}_{i} \otimes \ket{f_k}_{j}, \end{equation} where $\alpha_k$ are nonzero coefficients that satisfy $\sum_{k=1}^{3^M} |\alpha_k|^2 = 1$, $\ket{e_k}_{i}$ is an orthonormal basis of the $i$ black qutrits in the first $M$ collections and $\ket{f_k}_{j}$ is an orthonormal basis of the $j$ black qutrits in the first $M$ collections.
	
	We consider the following vectors in $p\calh_{code}$, which form a basis. Assume that $L \ge M$.
	\begin{align}
	\begin{split}  &\ket{L,k,k^\prime,p_{M+1},\cdots,p_L,q_{M+1},\cdots,q_L} = \\
	&\quad \ket{e_k}_{i} \otimes \ket{f_{k^\prime}}_{j} \otimes \left[\ket{p_{M+1}}_{i_{M+1}} \ket{q_{M+1}}_{j_{M+1}}\right]  \otimes \cdots \otimes \left[\ket{p_L}_{i_L} \ket{q_L}_{j_L}\right] \otimes \ket{\lambda} \cdots \label{eq:basisofcode}
	\end{split}
	\end{align} where $k$ and $k^\prime$ each label a basis vector for their respective black qutrits in the first $M$ collections, and $p_\ell$ and $q_\ell$ ($\ell \in \{ 1,2,\ldots,M\}$) each run over the three orthonormal basis vectors of their respective black qutrits in the $i$th collection. All black qutrit pairs past the $L$th collection are in the reference state $\ket{\lambda}$.
	
	We first consider the basis vectors that satisfy $L = M$. The vectors $\ket{M,k,k^\prime}$ and $\ket{M,\hat{k},\hat{k}^\prime}$ are orthogonal for $k^\prime \neq \hat{k}^\prime$. This is also true for the vectors $\calo \ket{M,k,k^\prime}$ and $\calo \ket{M,\hat{k},\hat{k}^\prime}$ since $\calo$ is a limit of operators which act as the identity on $\ket{f_{k^\prime}}$ in equation \eqref{eq:basisofcode}. Since $\sum_{k = 1}^{3^M} \alpha_{k}\ket{M,k,k} = \ket{\psi^\prime} \otimes \ket{\lambda}\cdots$, then $\calo  (\ket{\psi^\prime} \otimes \ket{\lambda} \cdots) = 0$ implies that $\calo \ket{M,k,k} = 0$ for all $k$. Let $\mathcal{U} \in A_{code}$ be an operator that acts as the identity operator on every vector in the tensor product in equation \eqref{eq:basisofcode} except that it may act arbitrarily on $\ket{f_{k^\prime}}$. We can choose $\mathcal{U}$ to send  $\ket{f_k}$ to $\ket{f_w}$ for $w \neq k$. Because $\mathcal{U}$ commutes with $\calo$, we have that $0 = \mathcal{U} \calo \ket{M,k,k} =  \calo \ket{M,k,w}$ and hence $\calo$ annihilates every basis vector with $L = M$. This argument can be repeated in a completely analogous way for the case $L > M$ (since $\ket{\psi^\prime} \otimes \ket{\lambda} \cdots = (\ket{\psi^\prime} \otimes \ket{\lambda}) \otimes \ket{\lambda} \cdots$ and $\ket{\psi^\prime} \otimes \ket{\lambda}$ has maximal Schmidt number) to show that $\calo$ annihilates all vectors in $p\calh_{code}$, and hence all of $\calh_{code}$.

\end{proof}


Recall that a vector is cyclic and separating for $M_{code}$ if and only if it is cyclic and separating for $M_{code}^\prime$ \cite{Witten:2018zxz}. Hence, cyclic and separating vectors for $M_{code}^\prime$ are also dense in $\calh_{code}$. 

\begin{thm}[\cite{HolographicEntropy}]
	If $\ket{\Psi} \in \calh_{code}$ is cyclic and separating with respect to $M_{code}$, then $u \ket{\Psi} \in \calh_{phys}$ is cyclic and separating with respect to $M_{phys}$.
\end{thm}
\begin{proof}
	To show that $u\ket{\Psi}$ is cyclic, we need to show that given any $\ket{\tilde{\Phi}} \in \calh_{phys}$ and $\epsilon > 0$, we can choose an operator $\mathcal{P} \in M_{phys}$ such that $||\mathcal{P} u \ket{\Psi} - \ket{\tilde{\Phi}}|| < \epsilon$.
	
 Choose $\ket{\tilde{\phi}} \in p\calh_{phys}$ such that $||\ket{\tilde{\phi}} - \ket{\tilde{\Phi}}|| < \frac{\epsilon}{2}$. Let $\ket{\tilde{\lambda}\cdots} \in p\calh_{phys}$ denote the vector for which all boundary qutrit pairs are in the reference state $\ket{\lambda}$. Choose an operator $\hat{\mathcal{P}} \in M_{phys}$ such that $\hat{\mathcal{P}} \ket{\tilde{\lambda}\cdots} = \ket{\tilde{\phi}}$. Choose $\calo \in M_{code}$ such that $||\calo \ket{\Psi} - \ket{\lambda\cdots}|| < \frac{\epsilon}{2 ||\hat{\mathcal{P}}||}$, where $\ket{\lambda \cdots} \in p\calh_{code}$ is the vector for which all qutrit pairs are in the reference state $\ket{\lambda}$. Let $\tilde{\calo}$ denote the image of $\calo$ under the tensor network map.
 
 Note that
\begin{equation} 
\ket{\tilde{\Phi}} - \hat{\mathcal{P}} \tilde{\calo} u \ket{\Psi} = 
\ket{\tilde{\Phi}}  - \ket{\tilde{\phi}}
 -  \hat{\mathcal{P}} u (\calo  \ket{\Psi} - \ket{\lambda\cdots})  .
  \end{equation}
Hence,
\begin{equation}
 ||\ket{\tilde{\Phi}} - \hat{\mathcal{P}} \tilde{\calo} u \ket{\Psi}|| \leq 
||\ket{\tilde{\Phi}}  - \ket{\tilde{\phi}}||
+||  \hat{\mathcal{P}} u (\calo  \ket{\Psi} - \ket{\lambda\cdots}) ||, 
\end{equation}
\begin{equation}
||\ket{\tilde{\Phi}} - \hat{\mathcal{P}} \tilde{\calo} u \ket{\Psi}|| \leq 
||\ket{\tilde{\Phi}}  - \ket{\tilde{\phi}}||
+||  \hat{\mathcal{P}}|| \cdot ||  (\calo  \ket{\Psi} - \ket{\lambda\cdots}) || ,
\end{equation}
\begin{equation}
 ||\ket{\tilde{\Phi}} - \hat{\mathcal{P}} \tilde{\calo} u \ket{\Psi}|| < \epsilon.
\end{equation}
We take $\mathcal{P} = \hat{\mathcal{P}} \tilde{\calo}$. This shows that $u\ket{\Psi}$ is cyclic with respect to $M_{phys}$. A completely analogous argument shows that $u\ket{\Psi}$ is cyclic with respect to $M_{phys}^\prime$, so it is also separating with respect to $M_{phys}$.
\end{proof}




\section{$M_{code}$ is a hyperfinite type II$_1$ factor}
\label{sec:type21}


In this section, we prove that $M_{code}$ satisfies the assumptions of Theorem \ref{thm:21theorem}, from which it follows that $M_{code}$ is a type II$_1$ factor. The same argument shows that $M_{code}^\prime$, $M_{phys}$, and $M_{phys}^\prime$ are also type II$_1$ factors. 

For $\calo \in M_{code}$, define the following linear function from $M_{code} \rightarrow \mathbb{C}$:
\begin{equation}
T( \calo ) := \braket{\lambda \cdots|\calo|\lambda \cdots},
\end{equation}
where $\ket{\lambda \cdots} \in \calh_{code}$ is the vector for which all pairs of black qutrits are in the state $\ket{\lambda}$. This clearly satisfies $ T(\calo^\dagger \calo) \ge 0$, $T( I ) = 1$, and $T(\calo^\dagger) = T^*(\calo)$.\footnote{The identity operator is denoted by $I$.}

For any operator $\calo_1 \in M_{code}$, it is possible to choose a neighborhood $\mathcal{N}$ of $\calo_1$ in the ultraweak operator topology such that $|T(\calo_2) - T(\calo_1)| < \epsilon$ for all $\calo_2 \in \mathcal{N} $. We may pick the neighborhood to be 
\begin{equation}
\mathcal{N} = \{\calo_2 \in M_{code} : |\braket{\lambda \cdots|(\calo_1-\calo_2)|\lambda \cdots}| < \epsilon \}.
\end{equation}
Hence, $T$ is ultraweakly continuous. 

For $a,b \in A_{code}$, it is easy to check that $T(ab) = T(ba)$. Since operators in $M_{code}$ may be written as strong limits of operators in $A_{code}$, $T(\calo_1 \calo_2) = T(\calo_2 \calo_1)\quad \forall\calo_1,\calo_2 \in M_{code}$.

 For $\calo \in M_{code}$, $T(\calo^\dagger \calo) = 0$ implies that $\calo = 0$ because $\ket{\lambda \cdots}$ is separating with respect to $M_{code}$. Hence, $T$ is faithful.

On a finite-dimensional Hilbert space $\calh$, any linear map $F : \calb(\calh) \rightarrow \mathbb{C}$ that satisfies $F(\calo_1 \calo_2) = F(\calo_2 \calo_1) \, \, \forall \calo_1,\, \calo_2 \in \calb(\calh)$ is proportional to the trace on $\calh$. It follows that for any linear map $\mathcal{T} : M_{code} \rightarrow \mathbb{C}$, $\mathcal{T}(a)$ for $a \in A_{code}$ is completely determined by the conditions such that $\mathcal{T}$ is linear, $\mathcal{T}(ab) = \mathcal{T}(ba)$ for $a,b \in A_{code}$, and $\mathcal{T}(I) = 1$. If $\mathcal{T}(a)$ is known for $a \in A_{code}$ and $\mathcal{T}$ is ultraweakly continuous, then the fact that $M_{code}$ is the strong closure of $A_{code}$ completely determines $\mathcal{T}(\calo)$ for all $\calo \in M_{code}$. Hence, $T$ is the only ultraweakly continuous normalized linear functional from $M_{code} \rightarrow \mathbb{C}$ that satisfies $T(\calo_1 \calo_2) = T(\calo_2 \calo_1)$  for all $\calo_1,\calo_2 \in M_{code}$.

Now, we may apply Theorem \ref{thm:21theorem}, where $\text{tr}(\calo) = T(\calo)$ for $\calo \in M_{code}$. Thus, $M_{code}$ is a type II$_1$ factor.

Recall that a von Neumann algebra $M$ is $\emph{hyperfinite}$ if $M = (\cup_n M_n)^{\prime \prime}$ where, for each $n \in \mathbb{N}$, each von Neumann subalgebra $M_n \subset M$ is finite-dimensional and $M_n \subset M_{n+1}$. The von Neumann algebra $M_{code}$ is hyperfinite because $M_{code} = A_{code}^{\prime \prime}$, and $A_{code} = \cup_N A_N$ where $A_N$ is the algebra of operators that can be written as $a^{(N)}$ in equation \eqref{eq:opinacode}. Each $A_N$ is a finite-dimensional algebra consisting of operators that act nontrivially on finitely many qutrits.


\subsection{More on the uniqueness of $T : M_{code} \rightarrow \mathbb{C}$}

Now, we explicitly show that for an ultraweakly continuous linear map $\mathcal{T} : M_{code} \rightarrow \mathbb{C}$, $\mathcal{T}(\calo)$ for $\calo \in M_{code}$ is completely determined given the value of $\mathcal{T}(a)$ for every $a \in A_{code}$.

The statement that $\mathcal{T}(\calo)$ is an ultraweakly continuous function of $\calo \in M_{code}$ implies that for any $\calo_1 \in M_{code}$ and any $\epsilon > 0$, there exists a neighborhood $\mathcal{N}$ of $\calo_1$ in the ultraweak operator topology such that for all operators $\calo_2 \in \mathcal{N}$, $|\mathcal{T}(\calo_2) - \mathcal{T}(\calo_1)| < \epsilon$. We may assume that $\mathcal{N}$ is given by
\begin{equation}
\mathcal{N} = \{\calo_2 \in M_{code}: \sum_{i = 1}^\infty|\braket{\eta_i|(\calo_1-\calo_2)|\xi_i}| < \epsilon \},\end{equation} 
for some $\epsilon > 0$ and some choice of sequences $\{\ket{\xi_i}\}$ and $\{\ket{\eta_i}\}$ satisfying 
\begin{equation} 
\sum_{i = 1}^\infty (||\ket{\xi_i}||^2+||\ket{\eta_i}||^2) < \infty. 
\end{equation}

Given $\calo_1 \in M_{code}$, let $\{a_n\} \in A_{code}$ be a sequence of operators that converges strongly to $\calo_1$. We need to show that for any choice of $\epsilon$ and $\{\ket{\xi_i}\},\{\ket{\eta_i}\}$, there exists an $N \in \mathbb{N}$ such that  $n > N \implies a_n \in \mathcal{N}$. We calculate
\begin{align}
\begin{split}
\sum_{i = 1}^\infty |\braket{\eta_i|(\calo_1-a_n)|\xi_i}| &= \sum_{i = 1}^{M-1} |\braket{\eta_i|(\calo_1-a_n)|\xi_i}| + \sum_{i = M}^\infty |\braket{\eta_i|(\calo_1-a_n)|\xi_i}|, \\
\sum_{i = M}^\infty |\braket{\eta_i|(\calo_1-a_n)|\xi_i}| &\leq \sum_{i = M}^\infty \frac{||\calo_1-a_n||}{2}(||\ket{\xi_i}||^2 + ||\ket{\eta_i}||^2) \leq K \sum_{i = M}^\infty (||\ket{\xi_i}||^2 + ||\ket{\eta_i}||^2),
\end{split}
\end{align}
for some $K >0$. We used the fact that the sequence of norms $\{||\calo_1-a_n||\}$ is bounded. First, choose $M$ so that 
\begin{equation}
K \sum_{i = M}^\infty (||\ket{\xi_i}||^2 + ||\ket{\eta_i}||^2) < \frac{\epsilon}{2}. 
\end{equation}Then, choose $N$ so that for all $n > N$, 
\begin{equation}
\sum_{i = 1}^{M-1} |\braket{\eta_i|(\calo_1-a_n)|\xi_i}| < \frac{\epsilon}{2}. 
\end{equation}
Hence for any $\epsilon > 0$, it is possible to choose an $N \in \mathbb{N}$ such that for $n > N$, $|\mathcal{T}(\calo_1) - \mathcal{T}(a_n)| < \epsilon$. Then we can conclude that
\begin{equation}
\lim_{n \rightarrow \infty} \mathcal{T}(a_n) = \mathcal{T}(\calo_1).
\end{equation}
If $\mathcal{T}(a)$ is known for all $a \in A_{code}$, then $\mathcal{T}(\calo)$ is known for all $\calo \in M_{code}$.

\section{The relative Tomita operator}
\label{sec:reltomita}

In this section, we study the relative Tomita operator defined on $\calh_{code}$.  See Section 3 of \cite{HolographicEntropy} for a review of Tomita-Takesaki theory. Given $\ket{\Psi},\ket{\Phi} \in \calh_{code}$, the relative Tomita operator with respect to $M_{code}$ is denoted by $S^c_{\Psi|\Phi}$. For $\calo \in M_{code}$,
\begin{equation} 
S^c_{\Psi|\Phi} \calo \ket{\Psi} = \calo^\dagger \ket{\Phi}.
\label{eq:codetomitadefinition} 
\end{equation} 
 The vector $\ket{\Psi}$ must be cyclic and separating with respect to $M_{code}$, but $\ket{\Phi}$ can be anything. In this section, we show that $S^c_{\Psi|\Phi}$ can be bounded or unbounded, depending on the choice of $\ket{\Psi}$ and $\ket{\Phi}$. In Section \ref{sec:normtomitafinite}, we compute the norm of the relative Tomita operator for a general, finite-dimensional Hilbert space. In Sections \ref{sec:tomitabounded} and \ref{sec:tomitaunbounded}, we provide one example in our setup where $S^c_{\Psi|\Phi}$ is bounded, and one example where $S^c_{\Psi|\Phi}$ is unbounded.

\subsection{Norm of the Tomita operator in a finite-dimensional Hilbert space}

\label{sec:normtomitafinite}

In this section we consider a Hilbert space $\calh = \calh_1 \otimes \calh_2$ for finite-dimensional Hilbert spaces $\calh_1$ and $\calh_2$ with equal dimension $D$. We want to compute the norm of the relative Tomita operator $S_{\Psi|\Phi}$ defined with respect to the algebra of operators acting on $\calh_1$. First, we perform Schmidt decompositions of $\ket{\Psi}$ and $\ket{\Phi}$:
\begin{align}
\ket{\Psi}= \sum_{k = 1}^{D} \alpha_k \ket{e_k} \otimes \ket{f_k} , \quad
\ket{\Phi}= \sum_{k=1}^{D} \beta_k \ket{g_k} \otimes \ket{h_k} ,
\label{eq:schmidt}
\end{align}
where $\ket{e_k}$ and $\ket{g_k}$ ($k \in \{ 1,2,\ldots,D\}$) are orthonormal bases of $\calh_1$ and $\ket{f_k}$ and $\ket{h_k}$ are orthonormal bases of $\calh_2$. All of the $\alpha_k$ coefficients must be nonzero. The action of $S_{\Psi|\Phi}$ on any normalized state is given by
\begin{equation} S_{\Psi|\Phi} \sum_{i=1,j=1}^D c_{ij}\ket{e_j} \otimes \ket{f_i} = \sum_{i=1,j=1,k=1}^D\frac{c_{ij}^*}{\alpha_i^*}  \beta_k \braket{e_j|g_k} \ket{e_i}\otimes \ket{h_k}, \end{equation}
where $\sum_{i=1,j=1}^D |c_{ij}|^2 = 1$. The norm of $S_{\Psi|\Phi}$ is found by maximizing the norm of the right hand side above with respect to the coefficients $c_{ij}$, subject to the normalization constraint. One finds that
\begin{equation} ||S_{\Psi|\Phi}|| = \frac{\max_{k = 1}^D |\beta_k|}{\min_{k=1}^D |\alpha_k|}. \end{equation}

\subsection{Example where $S^c_{\Psi|\Phi}$ is bounded}
\label{sec:tomitabounded}
In this section, we show that it is possible to choose states for which the relative Tomita operator is bounded. We consider as a special case $S_{\psi|\phi}^c$ for $\ket{\psi},\ket{\phi} \in p\calh_{code}$. Suppose that for $\ket{\psi}$ (resp. $\ket{\phi}$), the qutrit pairs in the $n_\psi$th (resp. $n_\phi$th) collection and beyond are in the reference state $\ket{\lambda}$. We note that there are many choices of $n_\psi$ and $n_\phi$, but our argument is independent of the choice we make. 

We consider a finite case of $n$ by letting $n = \max{(n_\Psi,n_\Phi)}$. By considering equation \eqref{eq:codetomitadefinition} for the case that $\calo$ can be written as $a^{(N)}$ in equation \eqref{eq:opinacode} with $N = n -1$, we may see how $S^c_{\psi|\phi}$ acts on any vector in $p\calh_{code}$ for which the qutrit pairs in the $n$th collection and beyond are in the reference state $\ket{\lambda}$. Let us temporarily restrict our attention to the $9^{n-1}$-dimensional Hilbert subspace spanned by these vectors, which may be written as $\calh_{i} \otimes \calh_{j}$, where $\calh_{i}$ and $\calh_{j}$ are the $3^{n-1}$-dimensional Hilbert spaces containing the states of the qutrits labeled by $i$ and $j$ respectively in $n-1$ copies of Figure \ref{fig:twotofourexample}. Doing the Schmidt decomposition as in equation \eqref{eq:schmidt} (where we set $D = 3^{n-1}$), we find that the maximum value of $||S^c_{\psi|\phi} \ket{\chi}||$ for a normalized vector $\ket{\chi} \in \calh_{i} \otimes \calh_{j}$ is \begin{equation}\frac{\max_{k=1}^{3^{n-1}} |\beta_k|}{\min_{k=1}^{3^{n-1}} |\alpha_k|}.\end{equation} 
It is crucial that none of the $\alpha_k$ coefficients vanish.

Let us now restrict our attention to the larger subspace of $p\calh_{code}$ where all qutrit pairs in the $(n+1)$th collection and beyond are in the reference state $\ket{\lambda}$. We want to do Schmidt decompositions of $\ket{\psi}$ and $\ket{\phi}$ in this $9^n$ dimensional Hilbert subspace.
 Let $\alpha_k$,$\beta_k$,$\ket{e_k}$,$\ket{g_k}$,$\ket{f_k}$,$\ket{h_k}$ for $k \in \{ 1,2,\ldots,3^{n-1}\}$ be defined as in equation \eqref{eq:schmidt} for the Schmidt decomposition in the $9^{n-1}$ dimensional subspace considered in the previous paragraph. Next, define
\begin{equation} \ket{\hat{e}_p} := \left\{ \begin{array}{cc}
	\ket{e_p} \otimes \ket{0}, & p = 1,\ldots,3^{n-1} \\ 
	\ket{e_{p - 3^{n-1}}} \otimes \ket{1}, & p = 3^{n-1}+1,\ldots,2\cdot3^{n-1} \\ 
	\ket{e_{p - 2\cdot3^{n-1}}} \otimes \ket{2}, & p = 2\cdot3^{n-1}+1,\ldots, 3^{n}
	\end{array}  \right.,\end{equation}  
where $\ket{0},\ket{1},\ket{2}$ are states of the $n$th black qutrit labeled $i$. The vectors $\ket{\hat{g}_p}$,$\ket{\hat{f}_p}$, and $\ket{\hat{h}_p}$ are defined analogously. Furthermore, define
\begin{equation} \hat{\alpha}_p := \left\{ \begin{array}{cc}
\frac{1}{\sqrt{3}}\alpha_p, & p = 1,\ldots,3^{n-1} \\ 
\frac{1}{\sqrt{3}}\alpha_{p - 3^{n-1}} , & p = 3^{n-1}+1,\ldots,2\cdot3^{n-1} \\ 
\frac{1}{\sqrt{3}}\alpha_{p - 2\cdot3^{n-1}}, & p = 2\cdot3^{n-1}+1,\ldots, 3^{n}
\end{array}  \right. .\end{equation} We define $\hat{\beta}_p$ analogously. The Schmidt decomposition is then given by
\begin{align}
\ket{\psi} &= \sum_{p = 1}^{3^n} \hat{\alpha}_p \ket{\hat{e}_p} \otimes \ket{\hat{f}_p} ,
\\
\ket{\phi} &= \sum_{p=1}^{3^n} \hat{\beta}_p \ket{\hat{g}_p} \otimes \ket{\hat{h}_p} .
\end{align}
If $\ket{\chi}$ is a normalized vector in the $9^n$ dimensional subspace, then the maximum value of $||S^c_{\psi|\phi} \ket{\chi}||$ is \begin{equation}\frac{\max_{p = 1}^{3^n} |\hat{\beta}_p|}{\min_{p = 1}^{3^n} |\hat{\alpha}_p|} = \frac{\max_{k = 1}^{3^{n-1}} |\beta_k|}{\min_{k = 1}^{3^{n-1}} |\alpha_k|}.\end{equation} Iterating the procedure of doing the Schmidt decompositions in larger subspaces of the code pre-Hilbert space, we see that for any vector $\ket{\eta} \in p\calh_{code}$, \begin{equation}||S^c_{\psi|\phi} \ket{\eta}|| \leq \frac{\max_{k = 1}^{3^{n-1}} |\beta_k|}{\min_{k = 1}^{3^{n-1}} |\alpha_k|} ||\ket{\eta}||.\end{equation}

Choose any $\ket{\Theta} \in \calh_{code}$. Let $\{\ket{\theta_\ell}\} \in p\calh_{code}$ be a sequence that converges to $\ket{\Theta}$. Define a sequence of operators $\{a_\ell\} \in A_{code}$ such that $\ket{\theta_\ell} = a_\ell \ket{\psi} \, \forall \ell \in \mathbb{N}$. Note that $a_\ell^\dagger \ket{\phi} = S^c_{\psi|\phi} \ket{\theta_\ell} \, \forall \ell \in \mathbb{N}$. For any $\ell,m \in \mathbb{N}$, we then have that
\begin{equation} ||(a_\ell^\dagger - a_m^\dagger) \ket{\phi}|| \leq \frac{\max_{k = 1}^{3^{n-1}} |\beta_k|}{\min_{k = 1}^{3^{n-1}}|\alpha_k|} || \ket{\theta_\ell} - \ket{\theta_m}||. \end{equation}
Hence, $\lim_{\ell \rightarrow \infty} a_\ell^\dagger \ket{\phi}$ exists. Thus, $S^c_{\psi|\phi}$ is a bounded operator defined on all of $\calh_{code}$.


\subsection{Example where $S^c_{\Psi|\Phi}$ is unbounded}
\label{sec:tomitaunbounded}
In this section, we show that for a particular choice of $\ket{\Psi},\ket{\Phi} \in \calh_{code}$,  $S^c_{\Psi|\Phi}$ is unbounded. Let $\ket{\Psi}$ be the vector for which all qutrit pairs are in the reference state $\ket{\lambda}$. $\ket{\Phi}$ will be constructed as a limit of a sequence of vectors $\{\ket{\phi_n}\} \in p\calh_{code}$. Let $\{\delta_i\}$ be a sequence of positive real numbers such that $\sum_{i = 1}^\infty \delta_i$ is finite. For $N \in \mathbb{N}$, let $\ket{e^N_a}$, $a \in \{ 1,2,\ldots,3^N$\}, denote an orthonormal basis vector of the qutrits labeled $i$ (see Figure \ref{fig:twotofourexample}) in the first $N$ collections. In particular,
\begin{equation} \ket{e^{1}_1} := \ket{0}_{i_1}, \  \ket{e^{1}_2} := \ket{1}_{i_1}, \ \ket{e^{1}_3} := \ket{2}_{i_1}. \end{equation} 
\begin{equation} \ket{e^{N}_a} := \left\{ \begin{array}{cc}
\ket{e^{N-1}_a}_{i_1 \cdots i_{N-1}} \otimes \ket{0}_{i_N}, & a = 1,\ldots,3^{N-1} \\ 
\ket{e^{N-1}_{a - 3^{N-1}}}_{i_1 \cdots i_{N-1}} \otimes \ket{1}_{i_N}, & a = 3^{N-1}+1,\ldots,2\cdot3^{N-1} \\ 
\ket{e^{N-1}_{a - 2\cdot3^{N-1}}}_{i_1 \cdots i_{N-1}} \otimes \ket{2}_{i_N}, & a = 2\cdot3^{N-1}+1,\ldots, 3^{N}
\end{array}  \right. .\end{equation}

Let $\ket{f^N_a}$, $a \in \{ 1,2,\ldots,3^N\}$, denote an orthonormal basis vector of the qutrits labeled $j$ in the first $N$ collections, defined in the same way as above. Each $\ket{\phi_n}$ is defined by
\begin{equation} \ket{\phi_n} := \sum_{a = 1}^{3^n} \sum_{b = 1}^{3^n} c^n_{ab} \ket{e^n_a} \ket{f^n_b} \otimes \ket{\lambda}\cdots, \end{equation}
where $c^n_{ab}$ is a $3^n \times 3^n$ matrix to be specified. The $\otimes \ket{\lambda} \cdots$ indicates that all black qutrit pairs in the $(n+1)$th collection and beyond are in the reference state $\ket{\lambda}$. Choose an arbitrary $x \in \mathbb{R}$ such that $x > 0$. Each $c^n_{ab}$ is defined by
\begin{equation} c^1_{ab} := \frac{1}{\sqrt{3}} \left(\begin{array}{ccc}
x & 0 & 0 \\ 
0 & x & 0 \\ 
0 & 0 & x
\end{array} \right)_{ab} + \left(\begin{array}{ccc}
\delta_1 & 0 & 0 \\ 
0 & 0 & 0 \\ 
0 & 0 & 0
\end{array}\right)_{ab} ,\end{equation}



\begin{equation} c^2_{ab} := \frac{1}{\sqrt{3}} \left(\begin{array}{ccc}
c^1 & 0_{3\times3} & 0_{3\times3} \\ 
0_{3\times3} & c^1 & 0_{3\times3} \\ 
0_{3\times3} & 0_{3\times3} & c^1
\end{array} \right)_{ab} + \left(\begin{array}{cc}
\delta_2 & 0_{1 \times 8} \\
0_{ 8 \times 1} & 0_{8 \times 8} \\
\end{array}\right)_{ab} ,\end{equation}

\begin{equation} \nonumber \vdots \end{equation}
\begin{equation} c^n_{ab} := \frac{1}{\sqrt{3}} \left(\begin{array}{ccc}
c^{n-1} & 0_{3^{n-1}\times3^{n-1}} & 0_{3^{n-1}\times3^{n-1}} \\ 
0_{3^{n-1}\times3^{n-1}} & c^{n-1} & 0_{3^{n-1}\times3^{n-1}} \\ 
0_{3^{n-1}\times3^{n-1}} & 0_{3^{n-1}\times3^{n-1}} & c^{n-1}
\end{array} \right)_{ab} + \left(\begin{array}{cc}
\delta_n & 0_{1 \times (3^n - 1)}  \\ 
0_{(3^n - 1) \times 1} & 0_{(3^n - 1) \times (3^n - 1)} \\
\end{array}\right)_{ab}. \end{equation}
Assuming $n > m$, we see that $||\phi_n - \phi_{m}|| \leq \sum_{i  ={m+1}}^{n} \delta_i$. Thus, $\ket{\Phi} := \lim_{n \rightarrow \infty} \ket{\phi_n}$ exists.

To demonstrate that $S^c_{\Psi|\Phi}$ is unbounded, we will construct a sequence of bounded operators $\{a_n \}\in A_{code}$ such that $ \lim_{n \rightarrow \infty} a_n \ket{\Psi}=0$ while $\lim_{n\rightarrow \infty}a_n^\dagger \ket{\Phi}$ does not converge. For $n \in \mathbb{N}$, define
\begin{equation} a_n := \epsilon_n \sqrt{3^n} (\ket{e_1^n}\bra{e_1^n}_{i_1 \cdots i_n} \otimes I_{j_1 \cdots j_n}) \otimes I \cdots, \end{equation}
where $\{\epsilon_n\}$ is a sequence of positive real numbers that we will specify later. Note that
\begin{equation} a_n \ket{\Psi} = \epsilon_n  (\ket{e_1^n}_{i_1 \cdots i_n} \otimes \ket{f_1^n}_{j_1 \cdots j_n}) \otimes \ket{\lambda}\cdots, \end{equation}
\begin{equation} ||a_n \ket{\Psi}|| = \epsilon_n. \end{equation}
Hence, $\lim_{n \rightarrow \infty} a_n\ket{\Psi} = 0$ when $\lim_{n\rightarrow \infty} \epsilon_n = 0$.

Next, we will consider the sequence $\{a_n^\dagger \ket{\Phi}\}$. Note that, for $n \in \mathbb{N}$,

\begin{equation} a_n^\dagger \ket{\phi_n} = \epsilon_n \sqrt{3^n} c^n_{11} \ket{e^n_1} \ket{f^n_1} \otimes \ket{\lambda}\cdots, \end{equation}
\begin{equation} ||a_n^\dagger \ket{\phi_n}|| = \epsilon_n \sqrt{3^n} c_{11}^n =  \epsilon_n (x + \sum_{i = 1}^n \sqrt{3^i} \delta_i). \end{equation}
One can verify that $||a_n^\dagger \ket{\phi_n}|| \leq ||a_n^\dagger \ket{\Phi}||$. Hence,
\begin{equation} ||a_n^\dagger \ket{\Phi}|| \geq \epsilon_n (x + \sum_{k = 1}^n \sqrt{3^k} \delta_k). \end{equation}
We may set $\delta_k = \frac{1}{k^2}$. Then $(x + \sum_{k = 1}^n \sqrt{3^k} \delta_k)$ grows without bound. We may choose $\epsilon_n$ to go to zero slowly enough so that $\epsilon_n (x + \sum_{k = 1}^n \sqrt{3^k} \delta_k)$ also grows without bound. Hence, $||a_n^\dagger \ket{\Phi}||$ grows without bound, so $S^c_{\Psi|\Phi}$ is an unbounded operator.

\section{Computing relative entropy for hyperfinite von Neumann algebras}

\label{sec:computerelentropy}

While the definition of relative entropy for infinite-dimensional von Neumann algebras is elegant, it is difficult to use in practice. To compute the relative entropy, one in principle needs to explicitly perform a spectral decomposition of the relative modular operator. However, because our setup involves hyperfinite von Neumann algebras, we can show that there is a more practical method to compute relative entropy. Recall that a hyperfinite von Neumann algebra $M$ may be written as $M = (\cup_{n=1}^\infty M_n)^{\prime \prime}$ where each $M_n$ denotes a finite-dimensional subalgebra of $M$ and $M_n \subset M_{n + 1} \, \forall n \in \mathbb{N}$. We will show that given a hyperfinite von Neumann algebra $M$ and two cyclic and separating vectors, the relative entropy of the two vectors may be computed by computing their relative entropy with respect to $M_n$ and then taking the limit $n \rightarrow \infty$. This result parallels the result of \cite{Araki}, but our explanation is better suited for studying our setup.\footnote{
	In particular, \cite{Araki} shows that the relative entropy of two linear functionals on a von Neumann algebra is a limit of relative entropies computed with respect to finite-dimensional subalgebras. However, we are more interested in the relative entropy of two vectors in the Hilbert space. Given a Hilbert space vector, we show how to compute a finite-dimensional density matrix. This allows us to express the infinite-dimensional relative entropy of two vectors as a limit of finite-dimensional entropies.}
 Computing the relative entropy with respect to $M_n$ intuitively amounts to performing a partial trace and using the finite-dimensional relative entropy formula on the reduced density matrices. In the next subsection, we precisely describe how to use the finite-dimensional relative entropy formula to compute the relative entropy defined with respect to a finite-dimensional subalgebra of a hyperfinite algebra. In particular, we will write the entropy in a form that is convenient for taking the limit $n \rightarrow \infty$. In section \ref{sec:monoton}, we review the monotonicity of relative entropy, which we use later. In section \ref{sec:proofoflimit}, we fully explain why the limit of finite-dimensional entropies equals the infinite-dimensional entropy.

\subsection{Defining relative entropy with respect to a finite-dimensional subalgebra}

The purpose of this section is to describe the relative entropy defined with respect to a finite-dimensional subalgebra of a hyperfinite algebra in a way that will be useful when we consider the limit of larger and larger subalgebras. Let $M$ be a hyperfinite von Neumann algebra on $\calh$, and let $M_n$ be a finite-dimensional subalgebra of $M$. Let $\ket{\Psi},\ket{\Phi} \in \calh$ be cyclic and separating with respect to $M$. Suppose that we want to compute the relative entropy of $\ket{\Phi}$ and $\ket{\Psi}$ with respect to $M_n$. Note that while $\ket{\Phi}$ and $\ket{\Psi}$ are separating with respect to $M_n$, they need not be cyclic. However, they may still be thought of as cyclic if we restrict our attention to subspaces of $\calh$ denoted by $\overline{M_n \ket{\Psi}}$ and $\overline{M_n \ket{\Phi}}$.
\begin{defn}
	Given a Hilbert space $\calh$, a von Neumann algebra $M \subset \calb(\calh)$, and a vector $\ket{\Psi} \in \calh$, let $\overline{M \ket{\Psi}}$ denote the closure of the set of vectors generated by acting on $\ket{\Psi}$ with all operators in $M$. That is,
	\[ \overline{M \ket{\Psi}} := \{\ket{\chi} \in \calh : \exists \{\calo_n\} \in M, \, \, \lim_{n \rightarrow \infty} \calo_n \ket{\Psi} = \ket{\chi} \}. \]
\end{defn}
We now explain how to compute the relative entropy of $\ket{\Psi}$ and $\ket{\Phi}$ with respect to $M_n$. First, the relative Tomita operator $S^n_{\Psi|\Phi}$ is defined to map $\calo \ket{\Psi}$ to $\calo^\dagger \ket{\Phi}$ for all $\calo \in M_n$. The Tomita operator should be viewed as a map between two different Hilbert spaces, $\overline{M_n \ket{\Psi}}$ and $\overline{M_n \ket{\Phi}}$. Since $M_n$ is finite-dimensional, $S^n_{\Psi|\Phi}$ is a bounded operator on $\overline{M_n \ket{\Psi}}$. The relative modular operator $\Delta^n_{\Psi|\Phi} = S^{n \,\dagger}_{\Psi|\Phi}S^n_{\Psi|\Phi}$ is a self-adjoint operator on $\overline{M_n \ket{\Psi}}$, and it may be defined to act as the identity operator on the orthogonal complement $\left(\overline{M_n \ket{\Psi}}\right)^\perp$. Then, the relative entropy is defined as \begin{equation}\label{eq:relentropy} \mathcal{S}_n =-\braket{\Psi|\log \Delta^n_{\Psi|\Phi} | \Psi}.\end{equation}

Equation \eqref{eq:relentropy} will appear again when we consider the limit of larger subalgebras. We now relate $\mathcal{S}_n$ to the more familiar finite-dimensional relative entropy formula. Because $M_n$ is a finite-dimensional von Neumann algebra that acts on the finite-dimensional Hilbert space $\overline{M_n \ket{\Psi}}$, we note that $\overline{M_n \ket{\Psi}}$ may be written as \cite{Harlow:2016fse}
\begin{equation}
\overline{M_n \ket{\Psi}} = \bigoplus_\alpha \left(\calh_{A_\alpha} \otimes \calh_{\bar{A}_\alpha}\right),
\label{eq:hilbertspacedecomposition}
\end{equation}
while $M_n$ may be written as
\begin{equation}
M_n = \{ \bigoplus_\alpha \left(\calo_{A_\alpha} \otimes I_{\bar{A}_\alpha}\right) : \calo_{A_\alpha} \in \calb(\calh_{A_\alpha}) \}. 
\end{equation}
Restricting our attention to $\overline{M_n \ket{\Psi}}$, the vector $\ket{\Psi}$ is cyclic and separating with respect to $M_n$. This implies that for each $\alpha$, $\dim \calh_{A_\alpha} = \dim \calh_{\bar{A}_\alpha}$ \cite{Witten:2018zxz}.

We now explain how to obtain a density matrix on $\overline{M_n \ket{\Psi}}$ from $\ket{\Psi}$. Intuitively, one simply needs to perform a partial trace on $\ket{\Psi}\bra{\Psi}$, since $\ket{\Psi} \in \overline{M_n \ket{\Psi}}$. However, we follow a different procedure that will also allow us to obtain a density matrix on $\overline{M_n \ket{\Psi}}$ from $\ket{\Phi}$, even though we might have that $\ket{\Phi} \notin \overline{M_n \ket{\Psi}}$. Let us define a linear map $T_\Psi : M_n \rightarrow \mathbb{C}$ such that $T_\Psi(\calo) = \braket{\Psi|\calo|\Psi} \, \forall \calo \in M_n$. The map $T_\Psi$ is positive. Assuming that $\ket{\Psi}$ is normalized, $T_\Psi(I) = 1$. The map $T_\Psi$ is also faithful because $\ket{\Psi}$ is separating with respect to $M_n$. If we restrict the domain of $T_\Psi$ to the set of operators in $M_n$ that annihilate $\calh_{A_\alpha} \otimes \calh_{\bar{A}_\alpha}$ for all $\alpha \neq 1$, then we can naturally define a hermitian, positive operator on $\calh_{A_1}$ as follows. Let $\ket{i}, \quad i \in \{1,2,\cdots,\dim \calh_{A_1} \}$ denote an orthonormal basis of $\calh_{A_1}$. Any operator in $\calb(\calh_{A_1})$ may be written as a linear combination of the operators $\ket{i}\bra{j} \quad \forall i,j \in \{1,2,\cdots,\dim \calh_{A_1} \}$. To treat $\ket{i}\bra{j}$ as an operator in $M_n$ that acts on all of $\overline{M_n \ket{\Psi}}$, we define $\ket{i}\bra{j}$ to act as the identity on $\calh_{\bar{A}_1}$ and to annihilate the subspaces $\calh_{A_\alpha} \otimes \calh_{\bar{A}_\alpha} \, \, \forall \alpha \neq 1$. Then, we define the operator $\rho^{(1)}_\Psi \in \calb(\calh_{A_1})$ by $ \braket{i|\rho^{(1)}_\Psi|j} = T_\Psi(\ket{j}\bra{i})$. We then extend the definition of $\rho^{(1)}_\Psi$ to an operator on $\calh_{A_1} \otimes \calh_{\bar{A}_1}$ by defining $\rho^{(1)}_\Psi$ to act as the identity on $\calh_{\bar{A}_1}$. In this way, we can define an operator $\rho^{(\alpha)}_\Psi$ acting on each $\calh_{A_\alpha} \otimes \calh_{\bar{A}_\alpha}$. Then, we define the density matrix $\rho_\Psi \in M_n$ to be the direct sum of all the $\rho_{\Psi}^{(\alpha)}$ for all values of $\alpha$. That is,
\begin{equation}
\rho_\Psi = \bigoplus_\alpha \rho_\Psi^{(\alpha)}.  
\end{equation} Note that $\sum_\alpha \text{Tr}_{A_\alpha} \rho_\Psi^{(\alpha)} = 1$ by construction and that $\rho_\Psi$ only depends on $\ket{\Psi}$ through the linear map $T_\Psi$. Also, $\ket{\Psi}$ must be a purification of $\rho_\Psi$ on $\overline{M_n \ket{\Psi}}$.

Even though $\ket{\Phi}$ is not necessarily in $\overline{M_n \ket{\Psi}}$, we can still define a density matrix $\rho_\Phi$ on $\overline{M_n \ket{\Psi}}$ with the linear map $T_\Phi$, which is defined analogously to $T_\Psi$. Let $\ket{\tilde{\Phi}} \in \overline{M_n \ket{\Psi}}$ be a purification of $\rho_\Phi$. We want to ask how $\ket{\Phi}$ is related to $\ket{\tilde{\Phi}}$. Note that $\braket{\Phi|\calo|\Phi} = \braket{\tilde{\Phi}|\calo|\tilde{\Phi}} \, \forall \calo \in M_n$. Define the linear map $U^\prime : \overline{M_n \ket{\Phi}} \rightarrow \overline{M_n \ket{\Psi}}$ such that $U^\prime \calo \ket{\Phi} = \calo \ket{\tilde{\Phi}} \, \forall \calo \in M_n$. Because $M_n$ is finite-dimensional, $U^\prime$ is a bounded operator, and $U^\prime$ has trivial kernel because $\ket{\Psi}$ is separating with respect to $M_n$. Because $||\calo \ket{\Phi}|| = ||\calo \ket{\tilde{\Phi}}|| \, \forall \calo \in M_n$, $U^\prime$ is an isometry. Because $U^\prime$ is invertible, $U^\prime$ satisfies $U^{\prime \, \dagger} U^\prime = I$, and from its definition we can see that $U^\prime$ commutes with all operators in $M_n$. Because $\ket{\tilde{\Phi}} = U^\prime \ket{\Phi}$, we see that the relative modular operator $\Delta^n_{\Psi|\Phi}$ defined at the beginning of this section equals the relative modular operator $\Delta^n_{\Psi|\tilde{\Phi}}$. Then, the relative entropy of $\ket{\Psi}$ and $\ket{\Phi}$ computed with respect to $M_n$ is given by \begin{equation}\mathcal{S}_n =-\braket{\Psi|\log \Delta^n_{\Psi|\tilde{\Phi}} | \Psi}.\end{equation} Since $\ket{\Psi}$ and $\ket{\tilde{\Phi}}$ are both vectors in the same finite-dimensional Hilbert space $\overline{M_n\ket{\Psi}}$, it is straightforward to see \cite{Witten:2018zxz} that $\mathcal{S}_n$, defined in equation \eqref{eq:relentropy}, is given by equation (A.21) of \cite{Harlow:2016fse} for $\rho = \rho_\Psi$, $\sigma = \rho_\Phi$, $M = M_n$, which is the finite-dimensional relative entropy formula.

The relative entropy defined with respect to $M_n$ of the vectors $\ket{\Psi}$ and $\ket{\Phi}$ only depends on $\ket{\Psi}$ and $\ket{\Phi}$ through the linear maps $T_\Psi$ and $T_\Phi$. As long as we can represent $M_n$ on a finite-dimensional Hilbert space with a cyclic and separating vector, we can decompose the Hilbert space as in \eqref{eq:hilbertspacedecomposition} (see \cite{Harlow:2016fse} for the details) and compute the relative entropies using $\rho_\Psi$ and $\rho_\Phi$, which are defined from $T_\Psi$ and $T_\Phi$.

Applying the above discussion to our tensor network model, we let $M_n \subset M_{code}$ be a finite-dimensional subalgebra of $M_{code}$ that consists of operators that act on the black qutrits labeled $i$ (see Figure \ref{fig:twotofourexample}) in the first $n$ collections. Let $\ket{\Psi},\ket{\Phi} \in \calh_{code}$ be cyclic and separating with respect to $M_{code}$. To compute the relative entropy with respect to $M_n$ of $\ket{\Psi}$ and $\ket{\Phi}$, we consider the action of $M_n$ on the Hilbert space associated with the first $n$ qutrit pairs. The relative entropy may be computed from the density matrices $\rho_\Psi$ and $\rho_\Phi$, which are constructed using the linear maps $T_\Psi$ and  $T_\Phi$. This intuitively amounts to performing a partial trace on $\ket{\Psi}\bra{\Psi}$ and $\ket{\Phi}\bra{\Phi}$ over all of $\calh_{code}$ except the Hilbert space of the first $n$ qutrits. In this subsection, we have shown that the result is equivalent to equation $\eqref{eq:relentropy}$. In the remainder of this section we will show that the infinite $n$ limit of equation $\eqref{eq:relentropy}$ yields the relative entropy of $\ket{\Psi}$ and $\ket{\Phi}$ with respect to $M_{code}$.

\subsection{Monotonicity of Relative Entropy}

\label{sec:monoton}

To show that the limit of finite-dimensional relative entropies equals the infinite-dimensional relative entropy, we use the monotonicity of relative entropy, which is nicely explained using a graph argument in \cite{Witten:2018zxz, Borchers}. 
However, our proof of the monotonicity of relative entropy is slightly different, as we do not assume that cyclic states remain cyclic after restricting the von Neumann algebra to a subalgebra. In the remainder of section \ref{sec:computerelentropy}, we make use of definitions and theorems given in \cite{HolographicEntropy}, such as the spectral theorem.

	Let $\calh$ be a separable Hilbert space with an orthonormal basis $\{\ket{e_i}\}$. Any $\ket{\chi} \in \calh$ may be written as
	\begin{equation}
	\label{eq:chi}
	\ket{\chi} = \sum_{i  =1}^\infty \ket{e_i} \braket{e_i|\chi}.
	\end{equation}
	Define the operator $K : \calh \rightarrow \calh$ as 
	\begin{equation}
	\label{eq:kchi}
	K \ket{\chi} := \sum_{i  =1}^\infty \ket{e_i} \braket{\chi|e_i}.
	\end{equation}
	The sum in equation \eqref{eq:kchi} is convergent because the sum in equation \eqref{eq:chi} is convergent.
The operator $K$ satisfies the following properties:
\begin{itemize}
	\item $K^2 = I$,
	\item $K \ket{\alpha \psi + \beta \chi} = \alpha^* K\ket{\psi} + \beta^* K\ket{\chi} \quad \forall \alpha,\beta \in \mathbb{C} \quad \forall \ket{\psi}, \ket{\chi} \in \calh$,
	\item Given a sequence $\{\ket{\psi_n} \} \in \calh$ and a vector $\ket{\psi} \in \calh$, $\lim_{n \rightarrow \infty} \ket{\psi_n} = \ket{\psi}$ if and only if $\lim_{n \rightarrow \infty} K\ket{\psi_n} = K\ket{\psi}$,
	\item $\braket{K \psi|K \chi} = \braket{\chi|\psi} \quad \forall \ket{\psi},\ket{\chi} \in \calh$,
	\item $\braket{ \psi|K |\chi} = \braket{ \chi | K |\psi} \quad \forall \ket{\psi},\ket{\chi} \in \calh$.
\end{itemize}

\begin{defn}
	Let $X$ be a linear operator on $\calh$. The \emph{graph} of $X$ is a subset of the Hilbert space $\calh \oplus \calh$, given by
	\[ \Gamma_X := \left\{ \left( \begin{array}{c}
	\ket{\psi} \\ 
	X \ket{\psi}
	\end{array}  \right) \in \calh \oplus \calh : \ket{\psi} \in D(X) \right\}. \]
\end{defn}
Let $S$ be a closed, densely defined, antilinear operator on $\calh$. Define $X := K S$. Note that $X^\dagger X = S^\dagger  S$ and that $X$ is a closed, densely defined, linear operator on $\calh$. The graph $\Gamma_X$ is thus a closed linear subspace of the Hilbert space $\calh \oplus \calh$. We define $\Pi_X$ to be the projection operator onto $\Gamma_X$, which satisfies $\Pi_X^2 = \Pi_X^\dagger = \Pi_X$. Since any vector in $\calh \oplus \calh$ can be represented as a column vector \begin{equation}\left(\begin{array}{c}
\ket{\psi} \\ 
\ket{\phi}
\end{array} \right) \quad \text{ for } \ket{\psi},\ket{\phi} \in \calh,\end{equation} we may represent $\Pi_X$ as a two by two matrix:
\begin{equation}
\Pi_X = \left(\begin{array}{cc}
p_{11} & p_{12} \\ 
p_{21} & p_{22}
\end{array} \right),
\end{equation}
where each $p_{ij}$ ($i,j \in \{1,2\}$) is a bounded linear operator on $\calh$ (since $\Pi_X$ is bounded). For any $\ket{\psi},\ket{\chi} \in \calh$, we have that $X(p_{11}\ket{\psi} + p_{12} \ket{\chi}) = p_{21} \ket{\psi} + p_{22} \ket{\chi}$. Hence,
\begin{equation}
X p_{11} = p_{21}, \quad X p_{12} = p_{22}.
\end{equation}
The condition $\Pi_X = \Pi_X^\dagger$ implies that $p_{ij}^\dagger = p_{ji} \quad \forall i,j \in \{1,2\}$, and the condition $\Pi_X^2 = \Pi_X$ implies that $\sum_{k = 1}^2 p_{ik} p_{kj} = p_{ij} \quad \forall i,j \in \{1,2\}$. With these relations, one may show that
\begin{equation}
p_{i1}(X^\dagger X + 1) p_{1j} = p_{ij} \quad \forall i,j \in \{1,2\},
\end{equation}
which implies that
\begin{equation}
p_{11}(X^\dagger X + 1)(p_{11} \ket{\psi} + p_{12} \ket{\chi}) = (p_{11} \ket{\psi} + p_{12} \ket{\chi}).
\end{equation}
Note that the domain of $X$ is given by
\begin{equation}
D(X) =  \{p_{11} \ket{\psi} + p_{12} \ket{\chi} :\ket{\psi},\ket{\chi} \in \calh\}.
\end{equation}
Because $D(X)$ is a dense subset of $\calh$, it follows that
\begin{equation}
p_{11} = (1 + X^\dagger X)^{-1}.
\end{equation}
Then, we see that
\begin{equation}
p_{21} = X (1 + X^\dagger X)^{-1}, \quad p_{12} = (1 + X^\dagger X)^{-1} X^\dagger, \quad  p_{22} = X(1 + X^\dagger X)^{-1} X^\dagger.
\end{equation}

In the following theorem, we study modular operators as opposed to relative modular operators. We will make an explicit connection to monotonicity of relative entropy later.

\begin{thm}
	\label{thm:mon1}
	Let $M$ be a von Neumann algebra that acts on a Hilbert space $\calh$. Let $\ket{\Psi} \in \calh$ be cyclic and separating with respect to $M$. Let $S_\Psi^M$ be the Tomita operator defined with respect to $M$ and $\ket{\Psi}$. Let $N$ be a von Neumann subalgebra of $M$ (we do not assume that $\ket{\Psi}$ is cyclic with respect to $N$). On the closed subspace $\overline{N \ket{\Psi}} \subset \calh$, let $S^N_\Psi$ be the Tomita operator defined with respect to $N$ and $\ket{\Psi}$. On the orthogonal complement $\overline{N \ket{\Psi}}^\perp \subset \calh$, let $S^N_\Psi = K$, where $K$ is given in equation \eqref{eq:kchi}. Then for all $\ket{\Phi} \in \overline{N \ket{\Psi}}$ and all $s > 0$, 
	\begin{equation} \nonumber
	\braket{\Phi|\frac{1}{s + (S_\Psi^M)^\dagger S_\Psi^M}|\Phi} \geq 	\braket{\Phi|\frac{1}{s + (S_\Psi^N)^\dagger S_\Psi^N}|\Phi}.
	\end{equation}  
\end{thm}
\begin{proof}
	Let $X^M = K S_\Psi^M$ and $X^N = K S_\Psi^N$. Let $\Gamma_{X^M} \subset \calh \oplus \calh$ and $\Gamma_{X^N} \subset \calh \oplus \calh$ be the graphs of $X^M$ and $X^N$ respectively, with projections $\Pi_{X^M}$ and $\Pi_{X^N}$. Let $\Pi_{N\Psi}$ denote the projection onto the closed subspace $\Pi_{N\Psi} (\calh \oplus \calh)$, which is defined by  
	\begin{equation}
	\Pi_{N\Psi} (\calh \oplus \calh) := \left\{
	\left(\begin{array}{c}
	\ket{\psi} \\ 
	\ket{\chi}
	\end{array} \right) \in \calh \oplus \calh: \ket{\psi},\ket{\chi} \in \overline{N \ket{\Psi}}
	\right\} .
	\end{equation}	
	Note that the closed subspace $\Gamma_{X^N} \cap \Pi_{N\Psi} (\calh \oplus \calh)$ is completely determined by the Tomita operator defined with respect to $\ket{\Psi}$ and $N$ on the subspace $\overline{N \ket{\Psi}}$. Because $N$ is a subalgebra of $M$, it follows that
	\begin{equation}
	\Gamma_{X^M} \supset (\Gamma_{X^N} \cap \Pi_{N\Psi} (\calh \oplus \calh)).
	\end{equation}
	The projection onto the closed subspace $(\Gamma_{X^N} \cap \Pi_{N\Psi} (\calh \oplus \calh))$ is given by $\Pi_{X^N} \Pi_{N \Psi} = \Pi_{N \Psi} \Pi_{X^N}$. It follows that
	\begin{equation}
	\Pi_{X^M} \geq \Pi_{X^N} \Pi_{N \Psi}.
	\end{equation} 
	If we evaluate the expectation value of the above equation in the state $\left(\begin{array}{c}
	\ket{\Phi} \\ 
	0
	\end{array} \right)$ for $\ket{\Phi} \in \overline{N \ket{\Psi}}$, we find that
	\begin{equation}
	\braket{\Phi|\frac{1}{1 + (X^M)^\dagger X^M}|\Phi} \geq 	\braket{\Phi|\frac{1}{1 + (X^N)^\dagger X^N}|\Phi},
	\end{equation}  
	which implies
	\begin{equation}
	\braket{\Phi|\frac{1}{1 + (S_\Psi^M)^\dagger S_\Psi^M}|\Phi} \geq 	\braket{\Phi|\frac{1}{1 + (S_\Psi^N)^\dagger S_\Psi^N}|\Phi}.
	\end{equation}
	By repeating the above logic with $X^M = \frac{1}{\sqrt{s}}K S_\Psi^M$ and $X^N = \frac{1}{\sqrt{s}}K S_\Psi^N$ for $s > 0$, we have that
	\begin{equation}
	\braket{\Phi|\frac{1}{s + (S_\Psi^M)^\dagger S_\Psi^M}|\Phi} \geq 	\braket{\Phi|\frac{1}{s + (S_\Psi^N)^\dagger S_\Psi^N}|\Phi}.
	\end{equation} 
\end{proof}

\begin{thm}
\label{thm:mod2}
	Let $\Delta_1, \Delta_2$ be operators on $\calh$ that are densely defined, closed, self-adjoint and positive. Assume that, for some $\ket{\Phi} \in D(\Delta_1) \cap D(\Delta_2)$ and all $s > 0$,
	\[ \braket{\Phi|\frac{1}{s + \Delta_1}|\Phi} \geq \braket{\Phi|\frac{1}{s + \Delta_2}|\Phi}. \]
	Also assume that $\braket{\Phi|\log \Delta_1 | \Phi}$ and $\braket{\Phi|\log \Delta_2 | \Phi}$ are finite. Then
	\[ -\braket{\Phi|\log \Delta_1|\Phi} \geq -\braket{\Phi|\log \Delta_2|\Phi}.  \]
\end{thm}
\begin{proof}
	Let $P^1_\Omega,P^2_\Omega$ denote the projection-valued measures associated with $\Delta_1,\Delta_2$. We use the spectral theorem\footnote{See \cite{HolographicEntropy} for an explanation of the notation.} to write
	\begin{equation}
	\label{eq:logd1}
	- \braket{\Phi|\log \Delta_1 | \Phi} = -\int_0^\infty  \log \lambda \, d(\braket{\Phi|P^1_\lambda|\Phi}) = \int_0^\infty  \int_0^\infty ds \left(\frac{1}{s + \lambda} - \frac{1}{s + 1}\right) d(\braket{\Phi|P^1_\lambda|\Phi}).
	\end{equation}
	By Fubini's Theorem (\cite{ReedSimon}, page 26), we may interchange the order of integration above if the following integral converges: 	
	\begin{align}
	\label{eq:absvalueint}
	\begin{split}
	\int_0^\infty  \int_0^\infty ds \left|\frac{1}{s + \lambda} - \frac{1}{s + 1}\right| d(\braket{\Phi|P^1_\lambda|\Phi}) &= \int_0^\infty  |\log \lambda| \,  d(\braket{\Phi|P^1_\lambda|\Phi}) \\
	&= \left|\int_0^1  \log \lambda \,  d(\braket{\Phi|P^1_\lambda|\Phi})\right| + \int_1^\infty  \log \lambda \,  d(\braket{\Phi|P^1_\lambda|\Phi}).
	\end{split}
	\end{align}
	Note that
	\begin{equation}
	0 \leq \int_1^\infty  \log \lambda \,  d(\braket{\Phi|P^1_\lambda|\Phi}) \leq \int_1^\infty  (\lambda - 1) \,  d(\braket{\Phi|P^1_\lambda|\Phi}) = \braket{\Phi|\Delta_1|\Phi} - \braket{\Phi|\Phi}, 
	\end{equation}
	which implies that
	\begin{equation}
	\int_1^\infty  \log \lambda \,  d(\braket{\Phi|P^1_\lambda|\Phi})
	\end{equation}
	is finite. Because $\braket{\Phi|\log \Delta_1|\Phi}$ is finite by assumption, it follows that
	\begin{equation}
	\int_0^1 \log \lambda \, d (\braket{\Phi|P^1_\lambda|\Phi})
	\end{equation}
	is finite. Thus, equation \eqref{eq:absvalueint} is finite, which implies that the integrals in equation \eqref{eq:logd1} may be interchanged. Thus,
	\begin{align}
	\begin{split}
	- \braket{\Phi|\log \Delta_1 | \Phi} &= \int_0^\infty ds \quad   \left(\braket{\Phi|\frac{1}{s + \Delta_1}|\Phi} - \frac{\braket{\Phi|\Phi}}{s + 1}\right) \geq \int_0^\infty ds \quad   \left(\braket{\Phi|\frac{1}{s + \Delta_2}|\Phi} - \frac{\braket{\Phi|\Phi}}{s + 1}\right) \\
	&= \int_0^\infty ds \int_0^\infty \left(\frac{1}{s + \lambda} - \frac{1}{s + 1}\right) d(\braket{\Phi|P^2_\lambda|\Phi}) .
	\end{split}
	\end{align}
	We may switch the order of integration above for the same reason as in equation \eqref{eq:logd1}. Thus,
	\begin{equation} -\braket{\Phi|\log \Delta_1|\Phi} \geq -\braket{\Phi|\log \Delta_2|\Phi}.  \end{equation}
\end{proof}

\subsection{The infinite-dimensional relative entropy as a limit of finite-dimensional relative entropies}

\label{sec:proofoflimit}

In this section, we use the above theorems to show how one could compute the relative entropy of two cyclic and separating vectors of a hyperfinite von Neumann algebra as a limit of finite-dimensional relative entropies.

\begin{thm}
	Let $M$ be a von Neumann algebra acting on a Hilbert space $\calh$ such that $M$ is generated by $\cup_{n = 1}^\infty M_n$, where $\{M_n\}$ is a sequence of finite-dimensional von Neumann subalgebras of $M$ satisfying $M_n \subset M_{n + 1} \, \, \forall n \in \mathbb{N}$. Let $\ket{\Psi},\ket{\Phi} \in \calh$ both be cyclic and separating with respect to $M$. Let $\cals_n$ denote the relative entropy of $\ket{\Psi}$ and $\ket{\Phi}$ defined with respect to $M_n$ (see equation \eqref{eq:relentropy} for details). Let $\cals$ denote the relative entropy defined with respect to $M$. Then
	\[ \lim_{n \rightarrow \infty} \cals_n  = \cals.\]
	In particular, if the limit does not converge, then $\cals$ is infinity.
	\end{thm}
\begin{proof}
	We mostly follow the logic of the proof of Lemma 3 of \cite{Araki}. We consider the tensor product Hilbert space $\calh \otimes \calk$, where $\calk$ is a four-dimensional Hilbert space spanned by orthonormal basis vectors $\ket{e_{ij}} \, ( i,j \in \{1,2\})$. Let $M_{2\times2}$ be a four-dimensional von Neumann algebra spanned by the operators $u_{ij} \, (i,j \in \{1,2\})$, which act on the basis vectors of $\calk$ as $u_{ij} \ket{e_{k\ell}} = \delta_{jk} \ket{e_{i \ell}}$. It follows that $u_{ij}^\dagger = u_{ji}$. Define $\hat{M} := M \otimes M_{2 \times 2}$ and $\hat{M}_n := M_n \otimes M_{2 \times 2}$. Let $$\ket{\hat{\Phi}} = \ket{\Phi} \otimes \ket{e_{11}} + \ket{\Psi} \otimes \ket{e_{22}}.$$ Note that $\ket{\hat{\Phi}}$ is cyclic and separating with respect to $\hat{M}$. Let $\hat{\Delta}$ denote the modular operator defined with respect to $\hat{M}$ and $\ket{\hat{\Phi}}$.  Let the operator $\hat{\Delta}_n$ act on $\overline{\hat{M}_n \ket{\hat{\Phi}}}$ as the modular operator defined with respect to $\ket{\hat{\Phi}}$ and $\hat{M}_n$, and let $\hat{\Delta}_n$ act as the identity on $\overline{\hat{M}_n \ket{\hat{\Phi}}}^\perp$. Note that 
	\begin{align}
	\begin{split}
	\hat{\Delta} (\ket{\Theta} \otimes \ket{e_{12}}) &= (\Delta_{\Psi | \Phi} \ket{\Theta}) \otimes \ket{e_{12}}, \quad \ket{\Theta} \in D(\Delta_{\Psi | \Phi}),
		\\
			(\log \hat{\Delta}) (\ket{\Theta} \otimes \ket{e_{12}}) &= ( (\log \Delta_{\Psi | \Phi}) \ket{\Theta}) \otimes \ket{e_{12}}, \quad \ket{\Theta} \in D(\log \Delta_{\Psi | \Phi}),
	\end{split}
	\end{align}  
	where $\Delta_{\Psi | \Phi}$ is the relative modular operator defined with respect to $M$, $\ket{\Psi}$, and $\ket{\Phi}$. We also have that
	\begin{align}
	\begin{split}
	\hat{\Delta}_n \ket{\Theta} \otimes \ket{e_{12}} &= (\Delta_{\Psi | \Phi}^n \ket{\Theta}) \otimes \ket{e_{12}}, \quad \ket{\Theta} \in \overline{M_n \ket{\Psi}},
	\\
		(\log \hat{\Delta}_n) (\ket{\Theta} \otimes \ket{e_{12}}) &= ( (\log \Delta^n_{\Psi | \Phi}) \ket{\Theta}) \otimes \ket{e_{12}}, \quad \ket{\Theta} \in \overline{M_n \ket{\Psi}},
	\end{split}
	\end{align}
	where $\Delta_{\Psi | \Phi}^n$ is the relative modular operator defined with respect to the finite-dimensional algebra $M_n$ (see the paragraph before equation \eqref{eq:relentropy} for details). Thus,
	\begin{align}
	\begin{split}
	\braket{u_{12}\hat{\Phi}|\log \hat{\Delta}|u_{12}\hat{\Phi}} &= \braket{\Psi|  \log \Delta_{\Psi | \Phi} | \Psi } = -\cals,
	\\
		\braket{u_{12}\hat{\Phi}|\log \hat{\Delta}_n|u_{12}\hat{\Phi}} &= \braket{\Psi|  \log \Delta^n_{\Psi | \Phi} | \Psi } = - \cals_n.
	\end{split}
	\end{align}
	Thus, we need to show that
	\begin{equation}
	\lim_{n \rightarrow \infty} \braket{u_{12}\hat{\Phi}|\log \hat{\Delta}_n|u_{12}\hat{\Phi}} = \braket{u_{12}\hat{\Phi}|\log \hat{\Delta}|u_{12}\hat{\Phi}}.
	\end{equation}
	Note that Theorems \ref{thm:mon1} and \ref{thm:mod2} imply relations between the finite-dimensional relative entropies $\cals_n$. That is, $\cals_n \leq \cals_{n +1} \quad \forall n \in \mathbb{N}$ because $M_n \subset M_{n+1}$. Also $\cals_n \leq \cals \, \forall n \in \mathbb{N}$. Thus, if $\lim_{n \rightarrow \infty} \cals_n$ does not converge, then $\cals$ must be infinity. For the remainder of the proof, we will thus assume that $\lim_{n \rightarrow \infty} \cals_n$ converges to a quantity that is less than or equal to $\cals$.

	Given the definitions of $\hat{\Delta}$ and $\hat{\Delta}_n$, it follows that (see \cite{Araki} and references therein)
	\begin{equation}
\label{eq:limgn}
	\lim_{n \rightarrow \infty} \braket{u_{12}\hat{\Phi}|g_N(\hat{\Delta}_n)|u_{12}\hat{\Phi}} = \braket{u_{12}\hat{\Phi}|g_N(\hat{\Delta})|u_{12}\hat{\Phi}},
	\end{equation}
	where $g_N(\lambda)$ is a continuous, bounded function on $\mathbb{R}$, defined for any $N \geq 1$, such that
	\begin{equation}
	g_N(\lambda) = \left\{ \begin{array}{cc}
	- \log N & \lambda \leq \frac{1}{N} \\ 
	\log \lambda & \frac{1}{N} \leq \lambda \leq N \\ 
	\log N & \lambda \geq N.
	\end{array} \right.
	\end{equation}
	Let $P_\Omega$ denote the spectral projections of $\hat{\Delta}$, and let $P^n_\Omega$ denote the spectral projections of $\hat{\Delta}_n$. By definition,
	\begin{equation}
\braket{u_{12}\hat{\Phi}|\log \hat{\Delta}|u_{12}\hat{\Phi}} = \int_0^\infty \log \lambda \, d(\braket{u_{12}\hat{\Phi}|P_\lambda|u_{12}\hat{\Phi}}).
	\end{equation}
Note that
	\begin{align}
	\begin{split}
	\braket{u_{12}\hat{\Phi}|g_N(\hat{\Delta})|u_{12}\hat{\Phi}} &= \int_0^\frac{1}{N} (- \log N)\, d(\braket{u_{12}\hat{\Phi}|P_\lambda|u_{12}\hat{\Phi}})\\
	 &+ \int_{\frac{1}{N}}^N \log \lambda \, d(\braket{u_{12}\hat{\Phi}|P_\lambda|u_{12}\hat{\Phi}}) + \int_N^\infty \log N \, d(\braket{u_{12}\hat{\Phi}|P_\lambda|u_{12}\hat{\Phi}}).
	\end{split}
	\end{align}	
Next, note that
\begin{align}
\begin{split}
0 &\leq \int_N^\infty (\log \lambda - \log N) \, d(\braket{u_{12}\hat{\Phi}|P_\lambda|u_{12}\hat{\Phi}}) = \int_N^\infty (\lambda \lambda^{-1}\log \frac{\lambda}{N}) \, d(\braket{u_{12}\hat{\Phi}|P_\lambda|u_{12}\hat{\Phi}})\\
&\leq (Ne)^{-1}\int_N^\infty (\lambda ) \, d(\braket{u_{12}\hat{\Phi}|P_\lambda|u_{12}\hat{\Phi}}) 
\\
&\leq (Ne)^{-1}\int_0^\infty (\lambda ) \, d(\braket{u_{12}\hat{\Phi}|P_\lambda|u_{12}\hat{\Phi}}) = (Ne)^{-1} \braket{u_{12}\hat{\Phi}|\hat{\Delta}|u_{12}\hat{\Phi}}.
\label{eq:ineq1}
\end{split}
\end{align}	
We have used the inequality $\lambda^{-1} \log \frac{\lambda}{N} \leq (Ne)^{-1}$. Note that $\braket{u_{12}\hat{\Phi}|\hat{\Delta}|u_{12}\hat{\Phi}} = \braket{\Psi|\Delta_{\Psi|\Phi}|\Psi} = \braket{\Phi|\Phi}$ is a finite quantity \cite{Witten:2018zxz}. Likewise, we have that
\begin{equation}
\label{eq:ineq2}
0 \leq \int_N^\infty (\log \lambda - \log N) \, d(\braket{u_{12}\hat{\Phi}|P^n_\lambda|u_{12}\hat{\Phi}}) \leq (Ne)^{-1} \braket{u_{12}\hat{\Phi}|\hat{\Delta}_n|u_{12}\hat{\Phi}} = (Ne)^{-1} \braket{\Phi|\Phi}.
\end{equation}
From equation \eqref{eq:limgn}, we have that
\begin{align}
\begin{split}
&\lim_{n \rightarrow \infty} \left[\braket{u_{12}\hat{\Phi}|\log \hat{\Delta}_n|u_{12}\hat{\Phi}}
-\int_N^\infty (\log \lambda - \log N) d(\braket{u_{12}\hat{\Phi}|P^n_\lambda|u_{12}\hat{\Phi}})\right.\\
&\quad\quad\quad\quad\quad\quad\quad\quad\quad\quad\quad\quad\left. -\int_0^{\frac{1}{N}} (\log \lambda + \log N) d(\braket{u_{12}\hat{\Phi}|P^n_\lambda|u_{12}\hat{\Phi}})
\right] 
\\
&\quad\quad = \int_0^\frac{1}{N} (- \log N)\, d(\braket{u_{12}\hat{\Phi}|P_\lambda|u_{12}\hat{\Phi}}) + \int_{\frac{1}{N}}^\infty \log \lambda \, d(\braket{u_{12}\hat{\Phi}|P_\lambda|u_{12}\hat{\Phi}}) \\
&\quad\quad\quad\quad\quad\quad\quad\quad\quad\quad\quad\quad - \int_N^\infty (\log \lambda -\log N) \, d(\braket{u_{12}\hat{\Phi}|P_\lambda|u_{12}\hat{\Phi}}).
\end{split}
\end{align}
Note that for $N \ge 1$, 
\begin{equation}
\int_0^\frac{1}{N} (- \log N)\, d(\braket{u_{12}\hat{\Phi}|P_\lambda|u_{12}\hat{\Phi}}) \leq 0 \quad \text{and} \quad \int_0^{\frac{1}{N}} (\log \lambda + \log N) d(\braket{u_{12}\hat{\Phi}|P^n_\lambda|u_{12}\hat{\Phi}}) \le 0.
\end{equation}
	Thus,	
\begin{align}
\label{eq:finiteN}
\begin{split}
&\lim_{n \rightarrow \infty} \left[\braket{u_{12}\hat{\Phi}|\log \hat{\Delta}_n|u_{12}\hat{\Phi}}
-\int_N^\infty (\log \lambda - \log N) d(\braket{u_{12}\hat{\Phi}|P^n_\lambda|u_{12}\hat{\Phi}})
\right] 
\\
&\quad\quad\leq \int_{\frac{1}{N}}^\infty \log \lambda \, d(\braket{u_{12}\hat{\Phi}|P_\lambda|u_{12}\hat{\Phi}}) - \int_N^\infty (\log \lambda-\log N )\, d(\braket{u_{12}\hat{\Phi}|P_\lambda|u_{12}\hat{\Phi}}).
\end{split}
\end{align}	
Note that \begin{equation}
\lim_{N \rightarrow \infty} \int_{\frac{1}{N}}^\infty \log \lambda \, d(\braket{u_{12}\hat{\Phi}|P_\lambda|u_{12}\hat{\Phi}}) = \braket{u_{12}\hat{\Phi}|\log \hat{\Delta}|u_{12}\hat{\Phi}}.
\end{equation}
Using equations \eqref{eq:ineq1} and \eqref{eq:ineq2}, we can take the large $N$ limit of equation \eqref{eq:finiteN} to obtain
\begin{equation}
\lim_{n \rightarrow \infty} \braket{u_{12}\hat{\Phi}|\log \hat{\Delta}_n|u_{12}\hat{\Phi}}
\leq \braket{u_{12}\hat{\Phi}|\log \hat{\Delta}|u_{12}\hat{\Phi}},
\end{equation}
which implies that
\begin{equation}
\lim_{n \rightarrow \infty} \cals_n
\geq \cals,
\end{equation}
	which implies that $\cals$ is finite. Using the monotonicity properties proved earlier, it follows that
	\begin{equation}
	\lim_{n \rightarrow \infty} \cals_n
	= \cals.
	\end{equation}
\end{proof}

\section{Conclusion and outlook}

	It is widely believed that entanglement in a holographic field theory encodes properties of the bulk spacetime. In particular, boundary states with semiclassical bulk duals must be highly entangled so that local operators in the bulk may be reconstructed from different subregions of the boundary \cite{Almheiri:2014lwa,Harlow:2018fse}. The Reeh-Schlieder theorem implies that generic boundary states are highly entangled. The implications of boundary entanglement for bulk reconstruction have been explicitly studied using tensor networks with a finite number of tensors \cite{Pastawski:2015qua,HaydenQi}, which necessarily involve finite-dimensional Hilbert spaces. However, various existing toy models are not well-suited to study the implications for AdS/CFT of the Reeh-Schlieder theorem, which is formulated in the continuum limit of quantum field theory with an infinite-dimensional Hilbert space. Our primary motivation is to construct a model of bulk reconstruction where the Reeh-Schlieder theorem is manifestly true. More precisely, we want to associate von Neumann algebras with boundary subregions so that cyclic and separating states with respect to these algebras are dense in the boundary Hilbert space. 







Since tensor networks and quantum error correction have proven to be useful tools in understanding AdS/CFT \cite{Harlow:2016fse, fixedarea, beyondtoymodels}, it is natural to generalize existing tensor network models to models with an infinite number of tensors. Our strategy for constructing an infinite-dimensional QECC is to first construct a QECC that relates a code pre-Hilbert space to a physical pre-Hilbert space. Tensor networks with a repeating pattern provide a natural way to do this.  The HaPPY Code \cite{Pastawski:2015qua} is a tensor network constructed from a pentagonal tiling of hyperbolic space with a natural AdS/CFT interpretation. We plan to apply our strategy to the HaPPY code, as the HaPPY tensor network can be naturally extended to an arbitrarily large size. The explicit example described in this paper uses multiple disconnected tensor networks, but it would be more satisfying to use a connected tensor network such as the HaPPY code. If we can generalize the HaPPY code to a QECC with infinite-dimensional Hilbert spaces, we will be able to construct a more accurate toy model of entanglement wedge reconstruction.

An important aspect of AdS/CFT that our toy model captures is that subregions in the boundary theory are associated with von Neumann algebras. In our example, we study type II$_1$ factors acting on both the bulk and boundary Hilbert spaces. However, the local operator algebras that arise in quantum field theory are generically of type III$_1$. \cite{Witten:2018zxz, Landsman-vNalg, gabbani}. Thus, it would be satisfying to have a toy model of entanglement wedge reconstruction where the von Neumann algebras are of type III$_1$.

Our infinite-dimensional QECC satisfies both statements in Theorem \ref{thm:maintheorem} \cite{HolographicEntropy}. The assumptions and statements in Theorem \ref{thm:maintheorem} are physically motivated by the Reeh-Schlieder Theorem \cite{RS} and previous work on error correction and AdS/CFT \cite{DongHarlowWall,Almheiri:2014lwa,Harlow:2016fse,Jafferis:2015del}. Toy models with infinite-dimensional Hilbert spaces should allow us to better understand the physics of entanglement wedge reconstruction and holographic relative entropy, including the role that the Reeh-Schlieder theorem plays.

In light of the fact that the equivalence between bulk and boundary relative entropies is only approximately correct at large $N$, approximate entanglement wedge reconstruction has been studied in \cite{Cotler} using finite-dimensional von Neumann algebras and universal recovery channels. It would be interesting to see if an appropriate generalization of the explicit formulas given for finite-dimensional entanglement wedge reconstruction can be checked in an infinite-dimensional toy example. In the future, we want to apply the study of infinite-dimensional von Neumann algebras to entanglement wedge reconstruction beyond the planar/semiclassical limit.

While our primary motivation has been to understand the bulk reconstruction in AdS/CFT, we note that infinite tensor networks may also be useful in studying two-dimensional conformal field theories. In the algebraic approach to 2d conformal field theory, every interval $I$ on the circle is assigned a von Neumann algebra $\cala$, and if $I_1 \subset I_2$ for two intervals $I_1$ and $I_2$, the associated algebras $\cala_1$ and $\cala_2$ satisfy $\cala_1 \subset \cala_2$. In the case of 2d chiral conformal field theory studied in \cite{gabbani}, each algebra is isomorphic to the unique hyperfinite type III$_1$ factor. Furthermore, note that there is also a unique hyperfinite type II$_1$ factor \cite{Landsman-vNalg}. In our setup, we use an infinite tensor network to characterize the type II$_1$ factor $M_{code}$ as a particular subalgebra of $M_{phys}$ on $\calh_{phys}$. If infinite tensor networks can relate the algebra associated with an interval to a subalgebra associated with a subinterval, they could be used to probe aspects of 2d conformal field theory. It would be interesting to see how infinite tensor networks could be related to quantities such as primary operator dimensions or three-point function coefficients.

\section*{Acknowledgments}
The authors are grateful to Kai Xu for discussions and Daniel Harlow and Temple He for helpful comments on this paper. M.J.K. is supported by a Sherman Fairchild Postdoctoral Fellowship. This material is based upon work supported by the U.S. Department of Energy, Office of Science, Office of High Energy Physics, under Award Number DE-SC0011632. D.K. would like to acknowledge a partial support from NSF grant PHY-1352084. 

\appendix

\section{An example of a strongly convergent sequence of operators}
Let us give a nontrivial example a strongly convergent sequence in $A_{code}$. First we will make some preliminary definitions. 

\subsection{Preliminary definitions}
Consider the Hilbert space of a single qutrit. Let $V(\theta) := \left(\begin{array}{ccc}
e^{i \theta} & 0 & 0 \\ 
0& e^{-i \theta} & 0 \\ 
0& 0 & 1
\end{array} \right)$ be a unitary operator defined on this Hilbert space. Let $\ket{\gamma}$ be any normalized state of a single qutrit. Then one may show that \begin{equation} \braket{\gamma|V(\theta)|\gamma} = 1 + z_\gamma (1 - \cos \theta) + i \hat{z}_\gamma \sin \theta \end{equation}
where $z_\gamma$ and $\hat{z}_\gamma$ are real numbers that depend on $\ket{\gamma}$ and satisfy $|z_\gamma| \leq 1$, $|\hat{z}_\gamma| \leq 1$.

Furthermore, consider the expectation value of the operator $V(\theta) \otimes I$ in a two-qutrit Hilbert space in the state $\ket{\lambda} = \frac{1}{\sqrt{3}}(\ket{00} + \ket{11} + \ket{22})$

\begin{equation} \braket{\lambda| V(\theta) \otimes I |\lambda} = \frac{2 \cos \theta + 1}{3} = 1 - \frac{2}{3}(1-\cos \theta). \end{equation}

\subsection{The example}
We will define a sequence of operators in $A_{code}$ that we wish to study. First, we define a sequence of angles $\{\theta_{n}\}$. We will specify the actual values of $\theta_n, \, n \in \mathbb{N}$ later. We define the sequence of operators $\{a_n\} \in A_{code}$ to be

\begin{equation} a_1 = \left[ V_{i_1}(\theta_1) \otimes I_{j_1} \right] \otimes I \cdots \end{equation}
\begin{equation} a_2 = \left[ V_{i_1}(\theta_1) \otimes I_{j_1} \right] \otimes \left[ V_{i_2}(\theta_2) \otimes I_{j_2} \right] \otimes I \cdots \end{equation}
\[ \vdots \]
\begin{equation} a_n = \left[ V_{i_1}(\theta_1)  \otimes I_{j_1} \right] \otimes \left[ V_{i_2}(\theta_2) \otimes I_{j_2} \right] \otimes \cdots \otimes \left[ V_{i_n}(\theta_n) \otimes I_{j_n} \right] \otimes I \cdots \end{equation}
\begin{equation} \vdots \end{equation}
Each square brackets contains the black qutrits associated with one collection in Figure \ref{fig:twotofourexample}. Each of the operators in the sequence is unitary, so they all are bounded and have unit norm. Now, we want to investigate the convergence of this sequence acting on a basis vector of $p\calh_{code}$, such as the one given in equation \eqref{eq:codeprehbasis}.  Define $\ket{\psi_n} := a_n \ket{M,\{p,q\}}$ for $n \in \mathbb{N}$. Choose $n,m \in \mathbb{N}$ with $m > n$. Then we have
\begin{equation} ||\ket{\psi_m} - \ket{\psi_n}||^2 = 2 - \braket{\psi_m|\psi_n} - \braket{\psi_n|\psi_m} = 2 - \braket{M,\{p,q\}|a_m^\dagger a_n|M,\{p,q\}} - \braket{M,\{p,q\}|a_n^\dagger a_m|M,\{p,q\}}.  \end{equation}

Note that
\begin{equation} a_n^\dagger a_m = \left[ I_{i_1}  \otimes I_{j_1} \right] \otimes \cdots\otimes \left[ I_{i_n}  \otimes I_{j_n} \right] \otimes \left[ V_{i_{n+1}}(\theta_{n+1})  \otimes I_{j_{n+1}} \right] \otimes \cdots \otimes \left[ V_{i_m}(\theta_m) \otimes I_{j_m} \right] \otimes I \cdots \end{equation}

\begin{equation} \braket{M,\{p,q\}|a_n^\dagger a_m|M,\{p,q\}} = \prod_{k = n+1}^m Y_k \end{equation}
\begin{equation} Y_k = \left\{ \begin{array}{cc}
\braket{p_k|V_{i_k}(\theta_k)|p_k} & k \leq M \\ 
\braket{\lambda \cdots|V_{i_k}(\theta_k) \otimes I_{j_k}|\lambda \cdots} & k > M
\end{array} \right. \end{equation}

Another way to write $Y_k$ is
\begin{equation} Y_k = 1 + x_k (1 - \cos \theta_k) + i y_k \sin \theta_k = r_k e^{i \phi_k} \end{equation}
where $x_k$ and $y_k$ are real numbers satisfying $|x_k| \leq 1$, $|y_k| \leq 1$. One may show that
\begin{equation} 1 - 2 |1 - \cos \theta_k| \leq r_k \leq 1 + (1 - \cos \theta_k)^2 + 2 |1 - \cos \theta_k| + \sin^2 \theta_k \end{equation}
and that, if $|\theta_k| < \frac{\pi}{2}$,
\begin{equation} |\phi_k| \leq \arctan \frac{|\sin \theta_k|}{1 - |1 - \cos \theta_k|}. \end{equation}

Up until now we did not specify the choice of angles $\theta_k$. Now, we make a choice. First, we choose an arbitrary $\eta \in \mathbb{R}$ such that $0 < \eta < 1$. We choose $\theta_k$ such that each $r_k$ satisfies
\begin{equation} 
e^{-\eta^k} < r_k < e^{\eta^k}
\end{equation}
and such that each $\phi_k$ satisfies
\begin{equation} 
-\eta^k < \phi_k < \eta^k. 
\end{equation}
We choose each $\theta_k$ to be nonzero. Thus,
\begin{equation} 
\prod_{k = n + 1}^m r_k < e^{\sum_{k = n+1}^m \eta^k} < e^{\frac{\eta^{n+1}}{1 - \eta}}, 
\end{equation}
\begin{equation} 
\prod_{k = n + 1}^m r_k > e^{-\sum_{k = n+1}^m \eta^k} > e^{-\frac{\eta^{n+1}}{1 - \eta}}, 
\end{equation}
\begin{equation} 
|\sum_{k = n +1}^m \phi_k| < \frac{\eta^{n+1}}{1 - \eta}. 
\end{equation}
We can therefore determine that the real part of $\prod_{k = n+1}^m Y_k$ is arbitrarily close to $1$ for $n$ sufficiently large. This means that $||\ket{\psi_m} - \ket{\psi_n}||$ is arbitrarily close to 0 for $n$ sufficiently large and $m > n$. This is enough to show that the sequence $\{\ket{\psi_n}\}$ is Cauchy, meaning that $\lim_{n \rightarrow \infty} a_n\ket{M,\{p,q\}}$ converges for every basis vector $\ket{M,\{p,q\}}$. Hence, $\lim_{n \rightarrow \infty} a_n \ket{\psi}$ converges for every $\ket{\psi} \in p\calh_{code}$. Since the sequence of norms $\{||a_n||\}$ is bounded from above (in particular, $||a_n||=1 \quad \forall n \in \mathbb{N}$), then the sequence of operators $\{a_n\}$ converges strongly.

\end{document}